\documentclass{amsart}

\usepackage{amsmath,amsthm}
\usepackage{amssymb,enumerate,url}
\usepackage{tikz}
\usetikzlibrary{backgrounds,calc,positioning}
\usepackage[linguistics]{forest}
\usetikzlibrary{arrows.meta}
\usepackage{booktabs,multirow} 
\usepackage{hyperref}

\usetikzlibrary{calc,arrows,decorations.pathreplacing}\usetikzlibrary{arrows,automata,positioning}
\tikzstyle{emptyKnot}=[scale=0.8]
\tikzstyle{knot}=[circle, fill,scale=0.4]
\tikzstyle{smallKnot}=[circle, fill,scale=0.3]

\newtheorem{lemma}{Lemma}

\newtheorem{theorem}{Theorem}
\newtheorem{example}{Example}

\theoremstyle{definition}

\newcommand{\G}{\mathcal{G}}
\newcommand{\N}{\mathbb{N}}
\newcommand{\T}{\mathcal{T}}
\newcommand{\cT}{\mathcal{C}}
\newcommand{\type}{\mathsf{type}}
\newcommand{\val}{\mathsf{val}}
\newcommand{\Prob}{\mathsf{Prob}}
\newcommand{\parent}{\operatorname{parent}}
\newcommand{\KL}{\,|\!|\,}
\newcommand{\fname}[1]{#1}
\newcommand{\ffcns}{\text{fcns}}

\allowdisplaybreaks

\begin{document}

\title{Entropy Bounds for Grammar-Based Tree Compressors}

\author[D. Hucke]{Danny Hucke}
\author[M. Lohrey]{Markus Lohrey}
\author[L. Seelbach Benkner]{Louisa Seelbach Benkner}
\address{Universit\"at Siegen, Germany}
\email{\{hucke,lohrey,seelbach\}@eti.uni-siegen.de}
\thanks{This work has been supported by the DFG research project
LO 748/10-1 (QUANT-KOMP)}

\begin{abstract}
The definition of $k^{th}$-order empirical entropy of strings is extended to node-labeled binary trees.
A suitable binary encoding of tree straight-line programs (that have been used for grammar-based tree compression before)
is shown to yield binary tree encodings of size bounded by the $k^{th}$-order empirical entropy plus some lower order 
terms. This generalizes recent results for grammar-based string compression to grammar-based tree compression.

\smallskip
\noindent \textbf{Keywords.} Grammar-based compression, binary trees, empirical entropy, lossless compression
\end{abstract}

\maketitle

\section{Introduction}

\paragraph{\bf Grammar-based string compression.}
The idea of grammar-based compression is based on the fact that in many cases a word $w$ can be succinctly
represented by a context-free grammar that produces exactly $w$. Such a grammar is called a {\em straight-line
program} (SLP) for $w$. In the best case, one gets an SLP of size $\Theta(\log n)$ for a word of length $n$,
where the size of an SLP is the total length of all right-hand sides of the rules of the grammar. A grammar-based
compressor is an algorithm that produces for a given word $w$ an SLP $\mathcal{G}_w$ for $w$, where, of course, 
$\mathcal{G}_w$
should be smaller than $w$.  Grammar-based compressors can be found at many places in the literature. Probably the best known example
is the classical {\sc LZ78}-compressor of Lempel and Ziv \cite{ZiLe78}. Indeed, it is straightforward to transform
the {\sc LZ78}-representation of a word $w$ into an SLP for $w$. Other well-known grammar-based compressors
are {\sc Bisection} \cite{KiefferYNC00}, {\sc Sequitur} \cite{Nevill-ManningW97}, and {\sc Repair} \cite{LaMo99}, just to mention a few.

Recently, several upper bounds on the compression perfomance of grammar-based compressors in terms of 
higher order empirical entropy have been shown. For this, the choice of a concrete binary encoding $B(\mathcal{G})$
of an SLP $\mathcal{G}$ is crucial.
Kieffer and Yang \cite{KiYa00} came up with such a binary encoding $B$ and proved that under certain assumptions
on the grammar-based compressor $w \mapsto \mathcal{G}_w$, the combined compressor 
$w \mapsto B(\mathcal{G}_w)$ yields a universal code with respect to the family of finite-state 
information sources over finite alphabets. More precisely, it is needed that the size of the SLP $\mathcal{G}_w$
is bounded by $\mathcal{O}(|w| / \log_{\hat\sigma}|w|)$ where $\sigma$ is the size of the underlying alphabet
and $\hat\sigma =\max\{2,\sigma\}$. This upper
bound is met by all grammar-based compressors that produce so-called irreducible SLPs \cite{KiYa00}, which
is the case for e.g. {\sc LZ78}, {\sc Bisection}, and {\sc Repair} after a small modification of the latter.
In their recent paper \cite{NaOch18}, Navarro and Ochoa used the binary
encoding $B(\mathcal{G}_w)$ from \cite{KiYa00} in order to prove for every word $w$ over an alphabet of size $\sigma$
the upper bound $|B(\mathcal{G}_w)| \le |w| H_k(w) + o(|w| \log\hat\sigma)$ for every $k \in o(\log_{\hat\sigma} |w|)$. Here,
$H_k(w)$ is the $k^{th}$-order empirical entropy of $w$, and the grammar-based compressor $w \mapsto \mathcal{G}_w$
must satisfy the upper bound $|\mathcal{G}_w| \le \mathcal{O}(|w| / \log_{\hat\sigma}|w|)$. Similar but weaker upper bounds for more practical
binary SLP-encodings have been shown in \cite{Gan18,NavarroR08}.

\medskip

\paragraph{\bf Grammar-based tree compression.}
Grammar-based compression has been generalized from strings to trees
by means of linear context-free tree grammars generating exactly one tree~\cite{BuLoMa07}.
Such grammars are also known as tree straight-line programs, TSLPs for short, see~\cite{Lohrey15dlt} for a survey.
TSLPs can be seen as a proper generalization of SLPs and DAGs (directed acyclic graphs, which are a widely
used compact representation of trees). Whereas DAGs only have the ability to share repeated subtrees of a tree,
TSLPs can also share repeated tree patterns with a hole (so-called contexts).
In \cite{GanardiHJLN17}, the authors presented a linear time algorithm that computes
for a given binary tree $t$ of size $n$
a TSLP $\G_t$ of size $\mathcal{O}(n / \log_{\hat\sigma} n)$ where $\sigma$ is the size of the underlying set of node labels
and $\hat\sigma =\max\{2,\sigma\}$.
An alternative algorithm with the same asymptotic size bound can be found in \cite{GanardiL17}.  TSLPs have been also extended
to so-called forest straight-line programs (FSLPs) which allow to compress unranked node-labeled trees \cite{GasconLMRS18}.
FSLPs are very similar to top DAGs \cite{BilleGLW15} and also meet 
the size bound $\mathcal{O}(n / \log_{\hat\sigma} n)$ for unranked trees
of size $n$.
The reader should notice that the $\mathcal{O}(n / \log_{\hat\sigma} n)$-bound
cannot be achieved by DAGs: the smallest DAG for an unlabeled binary tree of size $n$ may still contain 
$n$ edges. 

\medskip

\paragraph{\bf Entropy bounds for grammar-based tree compressors.}
In this paper we first consider node-labeled binary trees: every node has a label from a finite set
$\Sigma$ of size $\sigma$ and every non-leaf node has a left and a right child. 
For unlabeled binary trees the results of Kieffer and Yang 
on universal grammar-based compressors have been extended to trees in \cite{HuckeL17,ZhangYK14}. Whereas the universal tree encoder 
from \cite{ZhangYK14}  is based on DAGs (and needs a 
certain assumption on the average DAG size with respect to the input
distribution), the encoder from \cite{HuckeL17} uses TSLPs of size $\mathcal{O}(n / \log n)$. 
For this, a binary encoding of TSLPs similar to the one for SLPs from \cite{KiYa00} is proposed. 
In this paper we extend the binary TSLP-encoding from \cite{HuckeL17} to node-labeled 
binary trees and prove an entropy bound
similar to the one from \cite{NaOch18} for strings. To do this, we first have to come up with a reasonable
higher order entropy for binary node-labeled trees (we just speak of binary trees in the following). 
Several notions of tree entropy can be found in the literature,
but all are tailored towards unranked trees and do not yield nontrivial results for the special case of unlabeled binary trees.
\begin{itemize}
\item The $k^{th}$-order label entropy from \cite{FerraginaLMM05} is based on 
the empirical probability that a node $v$ is labeled with 
a certain symbol conditioned on the $k$ first labels from the parent node of $v$ to the root of the tree.
\item 
The tree entropy from \cite{JanssonSS12} is the $0^{th}$-order entropy of the node degrees.
\item Recently, two combinations of the two previous entropy measures were proposed in \cite{Ganczorz20}.
The first combination is based on the empirical probability that a node $v$ is labeled with 
a certain symbol conditioned on (i) the $k$ first labels from the parent node of $v$ to the root  and (ii) the node degree
of $v$.
The second combination uses the empirical probability that a node $v$ has a certain degree  
conditioned on (i) the $k$ first labels from the parent node of $v$ to the root and (ii) the node label of $v$.
\end{itemize}
Tree entropy \cite{JanssonSS12} is not useful in the context of binary trees, since
a binary tree with $n$ leaves has $n-1$ nodes of degree $2$, which shows that the
tree entropy divided by the number of nodes ($2n-1$) converges to $1$ when $n$ increases. 
On the other hand, the $k^{th}$-order label entropy \cite{FerraginaLMM05} is not useful for unlabeled trees.
For the special case of unlabeled binary trees, also the combinations of \cite{Ganczorz20} do not lead to
useful entropy measures.

Our first contribution is the definition of a reasonable entropy measure for binary trees that can be also used
for the unlabeled case.
For this we define the $k$-history of a node $v$ in a binary tree $t$ by taking the last $k$ edges on the unique 
path from the root to $v$. For each edge $(v_1, v_2)$ traversed on this path we write down the node label of $v_1$
and a $0$ (resp., $1$) if $v_2$ is a left (resp., right) child of $v_1$. Thus, the $k$-history of a node is a word of length
$2k$ that alternatingly consists of symbols from $\Sigma$ and directions that are encoded by $0$ or $1$.
For nodes at depth smaller than $k$ we pad the history with $0$'s and a default node label $\Box \in \Sigma$ in order to get length exactly $k$.\footnote{This is an ad hoc decision to make the definitions easier. In the appendix we discuss different approaches of how to deal with nodes of depth smaller than $k$, and prove
that they asymptotically lead to the same entropy measure.} For each $k$-history $h$ we then consider the joint probability distribution 
$P^t_h$ of the node degree (either $0$ or $2$) and the node label, conditioned on the history $h$.
Thus, $P^t_h(a,i)$ is the probability that a randomly chosen node among the nodes with history $h$ is labeled
with the symbol $a$ and has $i \in \{0,2\}$ children.
The $k^{th}$-order empirical entropy of $t$,
$H_k(t)$ for short, is then 
the sum of the entropies of these distributions $P_h^t$ (the sum is taken over all histories $h$)
weighted with the number of nodes with history $h$. This definition is similar to the definition of the $k^{th}$
order empirical entropy of a string.

Our main result states that 
\begin{equation} \label{entropy-bound}
|B(\mathcal{G}_t)| \le H_k(t) + \mathcal{O}(k n \log\hat\sigma / \log_{\hat\sigma} n) + \mathcal{O}(n \log \log_{\hat\sigma} n / \log_{\hat\sigma} n) + \sigma,
\end{equation}
where $t$ is a binary tree with $n$ leaves, the grammar-based compressor $t \mapsto \mathcal{G}_t$
produces TSLPs of size $\mathcal{O}(n/\log_{\hat\sigma} n)$ for binary trees of size $n$ with $\sigma$ many node labels and $\hat\sigma=\max(2,\sigma)$. Moreover, $B$ is an extension of the binary TSLP-encoding described in~\cite{HuckeL17} from unlabeled binary trees to labeled binary trees (Section \ref{sec-binary-coding}). If $k \leq o(\log_{\hat\sigma} n)$ then this bound can be 
simplified to $|B(\mathcal{G}_t)| \le H_k(t) + o(n \log \hat\sigma)$. The assumption $k \leq o(\log_{\hat\sigma} n)$ can be also found in \cite{NaOch18}. 
In fact, Gagie argued in \cite{Gagie06a} that the $k^{th}$-order empirical entropy for strings
stops being a reasonable complexity measure for almost all strings of length $n$ over alphabets of size $\sigma$  when 
$k \ge \log_{\hat\sigma} n$.

Our definition of $k^{th}$-order empirical entropy does not capture all regularities that can be exploited in grammar-based
compression: Take for instance a complete unlabeled binary tree $t_n$ of height $n$ (all paths from the root to a leaf have length $n$).
This tree has $2^n$ leaves and is very well compressible: its minimal DAG has only $n+1$ nodes, hence there also exists
a TSLP of size $n+1$ for $t_n$. But for every fixed $k$ the $k^{th}$-order empirical entropy of $t_n$ divided by $n$ converges to $2$ (the trivial
upper bound) for $n \to \infty$. If $n \gg k$ then for every $k$-history $z$ the number of leaves with $k$-history 
$z$ is roughly the same as the number of internal nodes with $k$-history $z$. Hence, although $t_n$ is highly compressible with TSLPs
(and even DAGs), its $k^{th}$-order empirical entropy is close to the maximal value. However, this phenomenon occurs for grammar-based string compression and 
the well-established higher-order empirical entropy of strings as well; see Section~\ref{sec-string-SLP}.

In Section~\ref{extension-unranked} we present a simple extension of our entropy notion to node-labeled 
unranked trees. In an unranked tree the number of children of a node is arbitrary. Unranked trees are important 
in the area of XML, where the hierarchical structure of a document is represented by a node-labeled unranked
tree. For such a tree $t$ we define the $k^{th}$-order empirical entropy as the $k^{th}$-order empirical entropy of 
the {\em first-child next-sibling} (fcns for short) encoding of $t$. The fcns-encoding of $t$ is a binary tree which 
contains all nodes of $t$. If a node $v$ of $t$ has the first (i.e., left-most) child $v_1$ and the right sibling $v_2$ then
$v_1$ (resp., $v_2$) is the left (resp., right) child of $v$ in the fcns-encoding of $t$. If $v$ has no child or no right sibling
then one adds dummy leaves to the fcns-encoding in order to obtain a full binary tree. 
Our choice of defining the $k^{th}$-order empirical entropy of an unranked tree via the fcns-encoding is motivated
by the fact that in XML document trees the label of a node $v$ usually depends on the labels of the ancestors and the labels
of the left siblings of $v$. This information is contained in the history of $v$ in the fcns-encoding.

We present experimental results with real XML document trees showing that in these cases the  
$k^{th}$-order empirical entropy is indeed very small compared to the worst-case bit size. 
An unranked tree with $n$ nodes and $\sigma$ node labels can be encoded with $2n +  \log_2(\sigma) n$
bits \cite{GearyRR06}. Up to low order terms, this is optimal.
Table~\ref{table} shows the values of 
the $k^{th}$-order empirical entropy (for $k=1,2,4,8$) divided by $2n +  \log_2(\sigma) n$ for several
real XML trees (that were also used in other experiments for XML compression \cite{LohreyMM13,LMR17}).
For $k = 4$, these quotients never exceed 20\% 
and for $k=8$ all quotients are bounded by 13.5\%. 

Our experimental results combined with our entropy bound  \eqref{entropy-bound} for grammar-based
compression are in accordance with the fact that grammar-based tree compressors yield excellent compression ratios
for XML document trees, see e.g. \cite{LohreyMM13}. Some of the XML documents from our experiments were 
also used in \cite{LohreyMM13}, where the performance of TreeRePair (currently the best grammar-based tree compressor from
a practical point of view) on XML document trees was tested. An interesting observation is that those XML trees, for which 
our $k$-th order empirical entropy is large are indeed those XML trees with the worst
compression ratio for TreeRePair in \cite{LohreyMM13} (this is in particular the Treebank document from Table~\ref{table}).

In a forthcoming paper we will compare  our definition of the $k^{th}$-order empirical entropy of trees with the above mentioned
tree entropies from \cite{FerraginaLMM05,Ganczorz20,JanssonSS12} for binary as well as unranked trees and both from a theoretical as well as experimental perspective.
A short version of this paper can be found in \cite{HuckeLS19}.

\section{Preliminaries}\label{sec-prelim}
In this section, we introduce some basic definitions concerning
information theory (Section~\ref{sec-empirical}) and binary trees 
(Section~\ref{sec-trees}). 
 
With $\N$ we denote the natural numbers including $0$.
We use the standard $\mathcal{O}$-notation. If $b>0$ is a constant, then
we just write $\mathcal{O}(\log n)$ for $\mathcal{O}(\log_b n)$. 
We make the convention that $0 \cdot \log(0) = 0$ and $0 \cdot \log(x/0)=0$ for $x \geq 0$.
For the unit interval $\{ r \in \mathbb{R} \mid 0 \leq r \leq 1 \}$
we write $[0,1]$. 

Let $w = a_1 a_2 \cdots a_l \in \Gamma^*$ be a word over an alphabet $\Gamma$. With $|w|=l$ we denote the 
length of $w$. The empty word is denoted
by $\varepsilon$. For $a \in \Gamma$ we denote with $|w|_a = |\{ i \mid 1 \leq i \leq l, a_i = a\}|$
the number of occurrences of $a$ in $w$.

\subsection{Empirical distributions and empirical entropy} \label{sec-empirical}

Let $A$ be a finite set. A probability distribution on $A$ is a mapping $p : A \to [0,1]$
such that $\sum_{a \in A} p(a) = 1$. For a probability distribution $p$ on $A$ 
we define its {\em Shannon entropy}
$$
H(p) = \sum_{a \in A} - p(a) \log_2 p(a) =  \sum_{a \in A} p(a) \log_2 (1/p(a)).
$$
We have $0 \le H(p) \le \log_2 |A|$.
A well-known generalization of Shannon's inequality states that for every
probability distribution $p$ on $A$ and any mapping $q : A \to [0,1]$
such that $\sum_{a \in A} q(a) \le 1$ we have
\begin{equation} \label{aczel}
H(p) =  \sum_{a \in A} - p(a) \log_2 p(a) \leq  \sum_{a \in A} - p(a) \log_2 q(a);
\end{equation}
see \cite{Aczel} for a proof. Shannon's inequality is the special case where 
$q$ is a probability distribution as well.
The Kullback-Leibler divergence between two probability distributions 
$p,q$ on $A$ (see \cite[Section~2.3]{CoTh06}) is defined as
\begin{equation} \label{def:D}
D(p \KL q) = \sum_{a \in A} p(a) \cdot \log_2(p(a)/q(a)) .
\end{equation}
It is known that $D(p \KL q) \ge 0$ 
for all $p,q$ (this follows from Shannon's inequality) 
and $D(p \KL q) = 0$ if and only if $p=q$.

Let $\overline{a} = (a_1, a_2, \ldots, a_l)$ be a tuple of elements that are from some (not necessarily finite) set $S$. 
The {\em empirical distribution} $p_{\overline{a}} : \{a_1, a_2, \ldots, a_l \} \to [0,1]$ of $\overline{a}$ is defined by
$$
p_{\overline{a}}(a) = \frac{|\{i \mid 1\le i\le l,\; a_i = a \}|}{n}  .
$$
We use this (and the following) definition also for words over some alphabet by identifying a word
$w = a_1 a_2 \cdots a_l$ with the tuple $(a_1, a_2, \ldots, a_l)$.
The {\em unnormalized empirical entropy} of $\overline{a}$ is 
\begin{equation}
\label{emp-entropy}
H(\overline{a}) = n \cdot H(p_{\overline{a}})  =  - \sum_{i=1}^l \log_2 p_{\overline{a}}(a_i) .
\end{equation}
From \eqref{aczel} it follows that for a tuple $\overline{a} = (a_1, a_2, \ldots, a_l)$ 
with $a_1, \ldots, a_l \in S$ 
and real numbers $q(a) \ge 0$ ($a \in S$) with 
$\sum_{a \in \{a_1, \ldots, a_l\}} q(a) \leq 1$ 
we have
\begin{equation} \label{shannon}
\sum_{i=1}^l - \log_2 p_{\overline{a}}(a_i)  \leq \sum_{i=1}^l - \log_2 q(a_i) .
\end{equation}
We also need the famous log-sum inequality, see e.g. \cite[Theorem~2.7.1]{CoTh06} (recall our conventions $0\cdot \log(0)=0$ and $0 \cdot \log(x/0)=0$ for $x \geq 0$):

\begin{lemma}\label{logsum}
Let $a_1, a_2, \dots, a_l,b_1, b_2, \dots, b_l \geq 0$ be real numbers. Moreover, let $a = \sum_{i=1}^l a_i$ and $b=\sum_{i=1}^l b_i$. Then
\begin{align*}
 a \log_2 \left(\frac{b}{a}\right) \geq \sum_{i=1}^l a_i \log_2\left(\frac{b_i}{a_i}\right).
\end{align*}
\end{lemma}

\subsection{Trees, tree processes, and tree entropy} \label{sec-trees}

\subsubsection{Trees and contexts}
Let $\Sigma$ denote a finite non-empty alphabet of size $|\Sigma| = \sigma$. 
Later, we will need a fixed distinguished symbol from $\Sigma$ that we will denote with $\Box \in \Sigma$.
We will also need the value $\hat{\sigma} = \max\{2,\sigma\}$.
With $\mathcal{T}(\Sigma)$ we denote the set of \emph{labeled binary trees} over the alphabet $\Sigma$. Formally, it is inductively defined as the smallest set of terms over $\Sigma$ such that 
\begin{itemize}
\item $\Sigma \subseteq \T (\Sigma)$ and 
\item if $t_1, t_2 \in \mathcal{T}(\Sigma)$ and $a \in \Sigma$, then $a(t_1, t_2) \in \T (\Sigma)$.
\end{itemize}
If e.g. $\Sigma=\{a,b\}$, then $a\in \mathcal{T}(\Sigma)$ is the binary tree with a single node labeled by $a$ and $a(b(b(a,b),a),a(b,a))\in \mathcal{T}(\Sigma)$ is the binary tree depicted on the left of Figure~\ref{fig:treeandcontext}.

A {\em tree encoder} is an injective mapping $E : \T(\Sigma) \to \{0,1\}^*$ such that the range $E(\T(\Sigma))$ is prefix-free, i.e.,
there do not exist $t, t' \in \T(\Sigma)$ with $t \neq t'$ such that $E(t)$ is a prefix of $E(t')$.

With $|t|$ we denote the number of leaves
of $t$, which can be inductively defined by
$|a|=1$ and $|a(t_1, t_2)| = |t_1|+|t_2|$ for $a \in \Sigma$ and $t_1, t_2 \in \mathcal{T}(\Sigma)$.
Note that $2|t|-1$ is the number of occurrences of symbols from $\Sigma$ in $t$. Let $\T_n (\Sigma) = \{ t \in \T(\Sigma) \mid |t| = n\}$ 
 for $n \geq 1$. Note that $\T_1 (\Sigma) = \Sigma$.
 We have $|\T_n (\Sigma)| = \sigma^{2n-1}C_{n-1}$,
where $C_k$ is the $k^{\text{th}}$ Catalan number. These numbers satisfy
the following well-known asymptotic estimate
\begin{equation} \label{catalan}
C_k \sim \frac{4^k}{\sqrt{\pi} k^{\frac{3}{2}}},
\end{equation}
see e.g. \cite{FlajoletS}. In fact, we have $C_k \leq 4^k$ for all $k \geq 0$ and hence
$|\T_n (\Sigma)| \le (2\sigma)^{2n}$.

A {\em context} is a labeled binary tree, where exactly one leaf is labeled with the special symbol $x \notin \Sigma$ (called the 
{\em parameter}); all other
nodes are labeled with symbols from $\Sigma$. Formally, the set of contexts $\cT (\Sigma)$ is the smallest set such that
\begin{itemize}
\item $x \in \cT (\Sigma)$ and 
\item if $a \in \Sigma$, $c \in \cT (\Sigma)$ and $t \in \T(\Sigma)$ then also $a(c,t), a(t,c) \in \cT (\Sigma)$.
\end{itemize}
If e.g. $\Sigma=\{a,b\}$, then $x\in \cT(\Sigma)$ is the context with a single node labeled by the parameter $x$ and $a(b(b(a,b),x),a(b,a))\in \mathcal{T}(\Sigma)$ is the context depicted on the right of Figure~\ref{fig:treeandcontext}.
For a tree or context $t\in\T (\Sigma) \cup \cT (\Sigma)$ and a context $c\in\cT (\Sigma)$,
we denote by $c[t]$ the tree or context which results from $c$ by replacing the unique occurrence of the parameter $x$ by $t$.
For example $c=a(a,x)$ and $t=b(a,a)$ yield $c[t]=a(a,b(a,a))$ (with $\Sigma = \{a,b\}$).
For a context $c$ we define $|c|$ inductively by $|x|=0$ and $|a(c,t)| = |a(t,c)| = |t|+|c|$ for 
$c \in \cT (\Sigma)$ and $t \in \T(\Sigma)$. In other words,
$|c|$ is the number of leaves of $c$, where the unique occurrence of the parameter $x$ is not counted.
Note that $|c| = |c[a]|-1$, where $a \in \Sigma$ is arbitrary. 
We define $\cT_n (\Sigma)= \{ c \in \cT (\Sigma) \mid |c| = n\}$ for $n \in \mathbb{N}$. Since the set 
$\Sigma$ will not change in this paper, we use the abbreviations $\mathcal{T}$,
$\mathcal{T}_n$, $\cT$, and $\cT_n$ for $\mathcal{T}(\Sigma)$,
$\mathcal{T}_n(\Sigma)$, $\cT (\Sigma)$, and $\cT_n(\Sigma)$, respectively.

Occasionally, we will consider a binary tree or context as a graph with nodes and edges in the usual way, where each node
is labeled with a symbol from $\Sigma$ (or $x$ in the case of a context).
Note that $t \in \T_n \cup \cT_n $ has $2n-1$ nodes in total: $n$ leaves and $n-1$ internal nodes.

It is convenient to define a node $v$ of $s \in \T \cup \cT$ as a bit string that describes the path from the root to the node ($0$ means left, $1$ means right).
Formally, we define the node set $V(s)\subseteq\{0,1\}^*$ of $s \in \T \cup \cT$ 
by 
\begin{itemize}
\item $V(a)=\{\varepsilon\}$ for every $a \in \Sigma$,
\item $V(x) = \emptyset$ and
\item $V( a(s_0, s_1)) = \{ i w \mid i \in \{0,1\}, w \in V(s_i) \}\cup \{\varepsilon\}$ for every $a \in \Sigma$.
\end{itemize}
Note that for a context $c\in\cT$, the set $V(c)$ does not contain the unique node in $c$ labeled with the parameter $x$.
We use this definition due to better readability of the paper since we mostly need the set of nodes without the parameter node.
Also, it is still possible to uniquely determine from $V(c)$ the path to the parameter $x$ due to the following properties:
For a tree $t\in \T$ we have $w0\in V(t)$ if and only if $w1\in V(t)$ for all $w\in\{0,1\}^*$ since each node has zero or two children.
The only context $c$ which fulfills this property is $c=x$, i.e., the parameter node is the only node of $c$ and $V(c)=\emptyset$.
For all other contexts $c\in\cT$ this property is violated since there exists a unique $w\in\{0,1\}^*$ such that $w0\in V(c)$ (respectively, $w1\in V(c)$) and $w1\notin V(c)$ (respectively, $w0\notin V(c)$). In this case the parameter node is $w1$ (respectively, $w0$).
Alternatively, the parameter node of a context $c$ is the single node in the set $V(c[a])\setminus V(c)$ for a symbol $a \in \Sigma$. We denote this
node with $\omega(c) \in \{0,1\}^*$. In other words: $V(c[a])\setminus V(c) = \{ \omega(c) \}$.

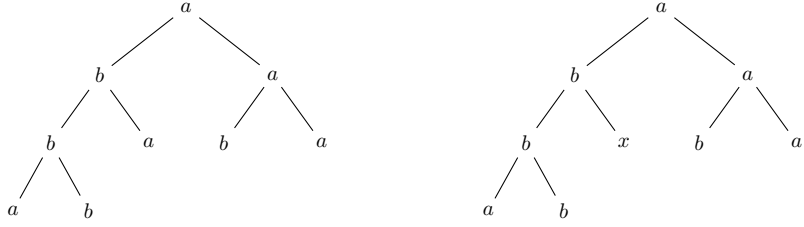
\begin{figure}
\begin{minipage}[hbt]{0.49\textwidth} 
\centering
\begin{tikzpicture}[-,level distance=9mm]
\tikzset{level 1/.style={sibling distance=23mm}}
\tikzset{level 2/.style={sibling distance=13mm}}
\tikzset{level 3/.style={sibling distance=10mm}}

\node [emptyKnot] (start){$a$}
  child {node [emptyKnot] (a) {$b$}
    child {node [emptyKnot] (c) {$b$}
      child {node[emptyKnot] (b1) {$a$}}
      child {node [emptyKnot] (h) {$b$}}
    }
    child {node[emptyKnot]  (f2) {$a$}}
  }
  child {node [emptyKnot] (b) {$a$}
     child {node[emptyKnot] (f3) {$b$}}
     child {node [emptyKnot] (e) {$a$}}
  }
;
\end{tikzpicture}
\end{minipage}
\begin{minipage}[hbt]{0.49\textwidth} 
\centering
\begin{tikzpicture}[-,level distance=9mm]
\tikzset{level 1/.style={sibling distance=23mm}}
\tikzset{level 2/.style={sibling distance=13mm}}
\tikzset{level 3/.style={sibling distance=10mm}}

\node [emptyKnot] (start){$a$}
  child {node [emptyKnot] (a) {$b$}
    child {node [emptyKnot] (c) {$b$}
      child {node[emptyKnot] (b1) {$a$}}
      child {node [emptyKnot] (h) {$b$}}
    }
    child {node[emptyKnot]  (f2) {$x$}}
  }
  child {node [emptyKnot] (b) {$a$}
     child {node[emptyKnot] (f3) {$b$}}
     child {node [emptyKnot] (e) {$a$}}
  }
;

\end{tikzpicture}
\end{minipage}
\caption{A tree (left) and a context (right).}
\label{fig:treeandcontext}
\end{figure}

\begin{example} \label{ex-tree}
Consider the tree $t=a(b(b(a,b),a),a(b,a))$ with $\Sigma=\{a,b\}$ depicted on the left of Figure~\ref{fig:treeandcontext}.We have $V(t)=\{\varepsilon, 0, 1, 00, 01, 10, 11, 000, 001\}$.
For the context $c=a(b(b(a,b),x),a(b,a))$ depicted on the right of Figure~\ref{fig:treeandcontext}, we have $t=c[a]$ and $\omega(c)=01$.
\end{example}

Consider a tree or context $s$ and let $v \in V(s)$.
The leaves of $s$ are those strings in $V(s)$ that are maximal with respect to the prefix relation. The length $|v|$ is the depth of the node $v$ in $s$
and the depth of $s$ is the maximal depth of a node in $V(s)$ (the depth of $s=x$ is not defined but also not needed).
Let $\lambda_s: V(s) \rightarrow \Sigma \times \{0,2\}$ denote the function mapping a node $v$ to the pair $(a,i)$ where $a \in \Sigma$ is the label of $v$ and $i \in \{0,2\}$ is the number of children of $v$. We can define this function inductively as follows:
\begin{itemize}
\item $\lambda_a(\varepsilon) = (a,0)$ for $a\in\Sigma$,
\item $\lambda_s(\varepsilon) = (a,2)$ for $s = a(s_0,s_1)$ with $a\in\Sigma$ and $s_0,s_1\in \T\cup\cT$,
\item $\lambda_s(iw) = \lambda_{s_i}(w)$ for $s = a(s_0,s_1)$ with $a\in\Sigma$, $s_0,s_1\in \T\cup\cT$ and $iw \in V(s)$.
\end{itemize}
Note that in the last case, if $s$ is a context, we cannot have $s_i = x$ because we must
have $w \in V(s_i)$.
In the following, we will omit the subscript $s$ in $\lambda_s(v)$ if $s$ is clear from the context. 

\subsubsection{Histories} \label{histories}
We now come to the crucial notion of the history of a node $v$ in a tree or context. Intuitively, the history of $v$ records
all information that can be obtained by walking from the root of the tree/context straight down to the node $v$.
First, we define the set of {\em histories} as
$$\mathcal{L} = (\Sigma  \{0,1\})^* =\{a_1i_1 \cdots a_ni_n \mid n \geq 0, a_k \in \Sigma, i_k \in \{0,1\} \text{ for all } 1 \leq k \leq n \}.$$ 
For an integer $k \geq 0$, let
$\mathcal{L}_k = \{w \in \mathcal{L} \mid |w| = 2k\}$ and let
$\ell_k: \mathcal{L} \rightarrow \mathcal{L}_k$ denote the partial function mapping a history $z \in \mathcal{L}$ 
with $|z| \geq 2k$ to the suffix of $z$ of length $2k$, i.e., $\ell_k(a_1i_1 \cdots a_ni_n) = a_{n-k+1} i_{n-k+1} \cdots a_ni_n$ (the function $\ell_0$ maps a string to the empty string).

For a tree $t$ and a node $v  \in V(t)$ (resp., 
a context $c$ and a node $v  \in V(c) \cup \{\omega(c)\}$), we inductively define its
\emph{history} $h(v) \in \mathcal{L}$ (in $t$) by 
\begin{itemize}
\item $h(\varepsilon) = \varepsilon$ and 
\item $h(wi) = h(w)ai$ for $i \in \{0,1\}$ and $wi \in V(t)$
(resp., $wi \in V(c) \cup \{ \omega(c)\}$). 
\end{itemize}
Here, $a$ is the symbol that labels the node $w$, i.e., 
$\lambda(w) = (a,2)$. That is, in order to obtain $h(v)$, while walking downwards in the tree from the root node to the 
node $v$ we alternately concatenate symbols from $\Sigma$ with 
binary numbers  in $ \{0, 1\}$ such that the symbol from $\Sigma$ 
corresponds to the label of the current node and the binary number $0$ (resp., $1$) states that we move on to the left (resp. right) child node.
Note that the symbol that labels $v$ is not part of the history of $v$.
The $k$-history of a tree node $v \in V(t)$ is 
$$
h_k(v) = \ell_k((\Box0)^{k} h(v)) \in \mathcal{L}_k, 
$$
i.e., the suffix of length $2k$ of the word $(\Box0)^k h(v)$, 
where $\Box$ is a fixed dummy symbol in $\Sigma$ (the choice is arbitrary). 
This means that if $|v|\ge k$ then $h_k(v)$ describes the last $k$ directions and node labels 
along the path from the root to node $v$. 
If $|v|< k$, we pad the history of $v$ with $\Box$'s and zeros such that $h_k(v) \in \mathcal{L}_k$.
In the appendix, we discuss other reasonable approaches of how to deal with nodes of depth smaller than $k$.
For $z \in \mathcal{L}_k$ we denote with 
$$V_z(t) = \{ v \in V(t) \mid h_k(v) = z\}$$ 
the set of nodes in $t$ with $k$-history $z$.

\begin{example} \label{ex-tree2}
Consider the tree $t=a(b(b(a,b),a),a(b,a))$ from Example~\ref{ex-tree} and let $\Box = a \in \Sigma$.
Then, $h(001) = h_3(001) = a 0 b 0 b 1$ and $h_4(10) = a0a0a1a0$.
\end{example}

\subsubsection{Tree processes}

A {\em tree process} is an infinite tuple $\mathcal{P} = (P_z)_{z \in \mathcal{L}}$ where 
every $P_z$ is a probability distribution on $\Sigma\times \{0,2\}$. With $\mathcal{P}$ we associate the function $\Prob_{\mathcal{P}} : \T \cup\cT \to [0,1]$
with
$$
\Prob_{\mathcal{P}}(s) = \prod_{v \in V(s)} P_{h(v)}(\lambda_s(v)) .
$$
We are mainly interested in this definition for the case that $s$ is a tree, but for technical reasons we also
have to allow contexts. Note that if $c$ is a context, then the parameter node of $c$ is not in $V(c)$ and 
therefore does not contribute to $\Prob_{\mathcal{P}}(c)$. 

A tree process can be used to randomly construct a tree from $\T$ as follows:
In a top-down way we determine for every tree node its label (from $\Sigma$) and its number of children,
where this decision depends on the history of the tree node. We start at the root node, whose history is the
empty word $\varepsilon$. If we have reached a tree node $v$ with history $z \in \mathcal{L}$ then we use the 
probability distribution $P_z$ to randomly choose a pair $(a,i) \in \Sigma\times \{0,2\}$. We assign the label $a \in \Sigma$
to $v$. If $i=0$ then $v$ becomes a leaf, otherwise the process continues at the two children $v0$ and $v1$ (whose history
is well-defined). Note that in this way we may produce infinite trees with non-zero probability (e.g. if $P_z(a,2) = 1$ for some
$a \in \Sigma$). Therefore, we only obtain an inequality instead of an equality in the following lemma (recall that $\T$
only contains finite trees).

\begin{lemma} \label{lemma-sum-trees}
Let $\mathcal{P}$ be a tree process. Then $\sum_{t \in \T} \Prob_{\mathcal{P}}(t) \leq 1$.
\end{lemma}

\begin{proof}
Define the set of trees $\T'_{n}$ inductively by $\T'_1 = \T_1$ and 
$$\T'_{n+1} = \T'_n \cup \{ a(t_1, t_2) \mid a \in \Sigma, t_1, t_2 \in \T'_n\}.$$
We have $\T'_n \subsetneq \T'_{n+1}$ and $\T = \bigcup_{n \ge 1} \T'_n$.
It then suffices to show $\sum_{t \in \T'_n} \Prob_{\mathcal{P}}(t) \leq 1$ for every $n \geq 1$.
This follows easily from the definition of $\Prob_{\mathcal{P}}(t)$ and the inductive definition of $\T'_n$.
\end{proof}
Lemma~\ref{lemma-sum-trees} cannot be extended to contexts, but the following bound will suffice for 
our purpose.

\begin{lemma} \label{lemma-sum-contexts}
Let $\mathcal{P}$ be a tree process. We have
$\sum_{c \in \cT_n} \Prob_{\mathcal{P}}(c) \leq n+1$ for every $n \geq 1$.
\end{lemma}

\begin{proof}
In order to bound $\sum_{c \in \cT_n} \Prob_{\mathcal{P}}(c)$, we first represent the probability of each context $c\in\cT_n$ as a sum of probabilities of trees.
So fix a context $c\in\cT_n$ for the first part of the proof.
Note first that in general no tree $t$ exists such that $\Prob_\mathcal{P}(c)\le \Prob_\mathcal{P}(t)$ (or even $\Prob_\mathcal{P}(c)=\Prob_\mathcal{P}(t)$) since $\omega(c)$ (the parameter node of $c$) does not contribute to the probability of the context $c$.
For example, the tree $c[a]$ ($a\in\Sigma$) which results from $c$ by replacing the parameter node by an $a$-labeled leaf node has probability $\Prob_\mathcal{P}(c)\cdot P_{h(\omega(c))}(a,0)\le \Prob_\mathcal{P}(c)$.
In order to bound $\Prob_{\mathcal{P}}(c)$, the idea is to replace the parameter node by all possible trees and not only by a single node.
So consider the set $c[\T]=\{c[t]\mid t\in \T \}$ of all trees that arise from $c$ by replacing the parameter by an arbitrary tree.
Unfortunately, the total probability $\sum_{t\in c[\T]}\Prob_\mathcal{P}(t)$ can still be strictly smaller than $\Prob_\mathcal{P}(c)$ since there might be infinite trees with positive probability with respect to $\mathcal{P}$.
To get rid of this problem, we fix an element $a \in \Sigma$ and modify $\mathcal{P}$ to a tree process $\mathcal{P}'=(P_z')_{z\in \mathcal{L}}$ such that (i) $P_z'=P_z$ for $|z|\le 2n$ and (ii) $P_z'(a,0)=1$ and $P_z'(a',i)=0$ for every $(a',i) \in \Sigma \times \{0,2\} \setminus \{(a,0)\}$ and $|z|>2n$.
The tree process $\mathcal{P}'$ is created such that all nodes $v$ of depth $|v|\le n$ contribute the probability $P_{h(v)}(\lambda(v))$ as before and all nodes of depth $n+1$ in a tree are $a$-labeled leaves with probability $1$.
Note first that for each context $c\in\cT_n$ and each node $v\in V(c)$ we have $|v|\le n$ and thus $P'_{h(v)}(\lambda(v))=P_{h(v)}(\lambda(v))$.
Secondly, all trees of depth larger than $n+1$ have probability $0$ with respect to $\mathcal{P}'$ (including infinite trees).
Hence, we get $\sum_{t\in\T}  \Prob_{\mathcal{P'}}(t)=1$.
We obtain
\begin{eqnarray*}
\sum_{t\in c[\T]}\Prob_\mathcal{P'}(t) & = & \sum_{t\in c[\T]} \prod_{v \in V(t)} P'_{h(v)}(\lambda(v))\\
&=&  \sum_{t\in c[\T]} \left(\prod_{v \in V(c)} P'_{h(v)}(\lambda(v)) \prod_{v \in V(t)\setminus V(c)} P'_{h(v)}(\lambda(v))\right)\\
&=& \Prob_\mathcal{P}(c) \cdot \underbrace{\sum_{t\in c[\T]}\prod_{v \in V(t)\setminus V(c)} P'_{h(v)}(\lambda(v))}_{(a)}.
\end{eqnarray*} 
We claim that $(a)$ equals $1$. To see this, consider the tree process $\mathcal{P''}=(P_z'')_{z \in \mathcal{L}}$ with $P_z''=P_{h(\omega(c))z}'$.
Also for $\mathcal{P''}$  only finite trees have non-zero probability and thus $\sum_{t\in\T }  \Prob_{\mathcal{P''}}(t)=1$.
We have
\begin{eqnarray*}
(a)&=& \sum_{t\in\T } \prod_{v \in V(t)} P'_{h(\omega(c))h(v)}(\lambda(v)) \\
&=& \sum_{t\in\T} \prod_{v \in V(t)}P_{h(v)}''(\lambda(v))\\
&=& \sum_{t\in\T}
\Prob_{\mathcal{P}''}(t)=1.
\end{eqnarray*}
It follows that $\Prob_\mathcal{P}(c)=\sum_{t\in c[\T]}\Prob_\mathcal{P'}(t)$.
In the second part of the proof it remains to bound $\sum_{c \in \cT_n} \Prob_{\mathcal{P}}(c)=\sum_{c \in \cT_n}\sum_{t\in c[\T]}\Prob_{\mathcal{P}'}(t)$.
The key point here is that for each tree $t\in\T $ there are at most $n+1$ different contexts $c\in\cT_n$ such that $t\in c[\T]$.
Note that for a tree $t$, the number of different contexts $c\in\cT_n$ such that $t\in c[\T]$ is exactly the number of nodes $v\in V(t)$ such that replacing the subtree rooted at $v$ by the parameter $x$ yields a context $c$ with $|c|=n$.
This is the same as the number of subtrees of $t$ with $|t|-n$ leaves.
Since different subtrees in $t$ of equal size do not share nodes, we can bound the number of subtrees with $|t|-n$ leaves by $|t|/(|t|-n)$.
We can assume that $|t|>n$ since otherwise there is no context $c\in \cT_n$ such that $t\in c[\T]$.
So we have $|t|=n+k$ for some $k>0$ and the number of subtrees of $t$ with $|t|-n$ leaves is at most $(n+k)/k=n/k+1\le n+1$.
We get
$$\sum_{c \in \cT_n}\sum_{t\in c[\T]}\Prob_{\mathcal{P}'}(t) \le (n+1) \sum_{t\in \T} \Prob_{\mathcal{P}'}(t) = n+1.$$
This concludes the proof of the lemma.
\end{proof}
A {\em $k^{th}$-order tree process} is a tree process $\mathcal{P} = (P_z)_{z \in \mathcal{L}}$ such that
$P_{z} = P_{z'}$ if $\ell_k((\Box0)^{k}z) = \ell_k((\Box0)^{k}z')$. Thus, the probability distribution that is chosen for a certain
tree node depends only on the $2k$ last symbols of the history of the node (where histories are padded with $\Box0$ on the left
to reach length $2k$ for the fixed symbol $\Box \in \Sigma$). 
We will identify the $k^{th}$-order tree process $\mathcal{P} = (P_z)_{z \in \mathcal{L}}$ 
with the finite tuple $(P_z)_{z \in \mathcal{L}_k}$; it contains all information about $\mathcal{P}$.
Note that for a $k^{th}$-order tree process $\mathcal{P}$ we can compute $\Prob_{\mathcal{P}}(s)$ for a tree or context $s$ as
\begin{equation} \label{eq:Prob-k-th-order}
\Prob_{\mathcal{P}}(s) = \prod_{z \in \mathcal{L}_k} \prod_{v \in V_z(s)} P_z(\lambda(v)),
\end{equation}
where the empty product (which arises in case $V_z(s) = \emptyset$) is $1$.

\subsubsection{Higher-order entropy of a tree} \label{sec:tree-entropy}

Let us fix $k \geq 0$.
We define the {\em $k^{th}$-order (unnormalized) empirical entropy} $H_k(t)$ of a tree $t \in \T_n$ as follows:
For $z \in \mathcal{L}_k$ let 
$$m^t_z = |V_z(t)|$$ 
be the number
of nodes of $t$ with $k$-history $z$
and for $\tilde{a} \in \Sigma \times \{0,2\}$ let 
\begin{equation} \label{def:m_za}
m^t_{z,\tilde{a}} = |\{ v \in V_z(t) \mid \lambda(v)=\tilde{a}\}|.
\end{equation}
We then define the {\em empirical $k^{th}$-order tree process}  $\mathcal{P}^t = (P^t_z)_{z \in \mathcal{L}_k}$
by
\begin{equation} \label{def:P_za}
P^t_z(\tilde{a}) = \frac{m^t_{z,\tilde{a}}}{m^t_z}
\end{equation}
for all $\tilde{a} \in \Sigma \times \{0,2\}$ and all
$z \in \mathcal{L}_k$ with $m^t_z > 0$.
If $m^t_z=0$, then we can define $P^t_z$ as an arbitrary distribution.
Then
\begin{equation} \label{def:H_k}
H_k(t) = \sum_{z \in \mathcal{L}_k} m^t_z H(P^t_z).                     
\end{equation}
Note that $$0 \le H_k(t) \leq (2n-1)\log_2(2\sigma) = (2n-1)(1+\log_2\sigma)$$ since $0 \le H(P^t_z) \leq \log_2(2\sigma)$ 
and $\sum_{z \in \mathcal{L}_k} m^t_z = 2n-1$. This upper bound on the entropy matches the information theoretic bound
for the worst-case output length of any tree encoder on $\T_n$.
Using the asymptotic bound \eqref{catalan} for the Catalan numbers, one sees that for any tree encoder there must exist
a tree $t \in \T_n$ which is encoded with $2 \log_2(2\sigma)n - o(n) = 2(\log_2\sigma+1)n -o(n)$ bits.
The $k^{th}$-order empirical entropy $H_k(t)$ is a lower bound on the coding length
of a tree encoder that encodes for each node the relevant information (the label of the node and the binary information
whether the node is a leaf or internal) depending on the $k$-history of the node.

\begin{example}
Let $t$ denote the binary tree $t= a(b(b(a,b),a),a(b,a))$ as depicted on the left of Figure\ref{fig:treeandcontext}. In order to compute the first order empirical entropy $H_1(t)$ of $t$, we have to consider $k$-histories of $t$ with $k=1$: Let $\Box=a$.
It follows that
 $V_{a0}(t)=\{\varepsilon,0,10\}$, $V_{b0}(t)=\{00,000\}$, $V_{a1}(t)=\{1,11\}$ and $V_{b1}(t)=\{01,001\}$. Thus, we have $m_{a0}^t=3$ and $m_{a1}^t=m_{b0}^t=m_{b1}^t=2$.
Next, for each $k$-history $z$, we consider $\lambda(v)$ for $v \in V_z(t)$: For $z=a0$, we have $\lambda(\varepsilon)=(a,2)$, $\lambda(0)=(b,2)$ and $\lambda(10)=(b,0)$. Hence, $m_{a0,(a,2)}^t=m_{a0,(b,0)}^t=m_{a0,(b,2)}^t=1$ and $H(P_{a0}^t)=\log_2(3)$. 
Analogously, we find $H(P_{b0}^t)=H(P_{a1}^t)=H(P_{b1}^t)=1/2\log_2(2)+1/2\log_2(2)=1$. Altogether, this yields $H_1(t)= 3\cdot \log_2(3) + 2\cdot 1 + 2 \cdot 1 + 2 \cdot 1$ which is roughly $9.3$.
\end{example}
One can define $H_k(t)$ alternatively in the following way:
Take a $k$-history  $z\in \mathcal{L}_k$ and enumerate the set $V_z(t)$ in an arbitrary way as
$v_1, v_2, \ldots, v_j$. Define the string $w(t,z) = \lambda(v_1) \lambda(v_2) \cdots \lambda(v_j) \in (\Sigma \times \{0,2\})^*$.
We have
$$
H_k(t) = \sum_{z \in \mathcal{L}_k} H(w(t,z)), 
$$
where the empirical entropy $H(w(t,z))$ is defined according to \eqref{emp-entropy}.

The following lemma and its proof are very similar to a corresponding statement for the $k^{th}$-order empirical
entropy of strings, see \cite{Gagie06a}. 

\begin{theorem} \label{theo-travis-for-trees}
Let $t \in \T $.
For every $k^{th}$-order tree process $\mathcal{P} = (P_z)_{z \in \mathcal{L}_k}$ with $\Prob_{\mathcal{P}}(t) > 0$
we have $$H_k(t) \leq - \log_2 \Prob_{\mathcal{P}}(t)$$ with equality
if and only if $P^t_z = P_z$ for all $z \in \mathcal{L}_k$ with $m^t_z > 0$.
\end{theorem}

\begin{proof}
We have
\begin{eqnarray*}
- \log_2 \Prob_{\mathcal{P}}(t) & \stackrel{\text{\eqref{eq:Prob-k-th-order}}}{=} & \sum_{z \in \mathcal{L}_k} \sum_{v \in V_z(t)} \log_2(1/P_z(\lambda(v)) \\
& \stackrel{\text{\eqref{def:m_za}}}{=} &  \sum_{z \in \mathcal{L}_k} \sum_{\tilde{a} \in \Sigma \times \{0,2\}} m^t_{z,\tilde{a}} \log_2(1/P_z(\tilde{a})) \\
& \stackrel{\text{\eqref{def:P_za}}}{=} &  \sum_{z \in \mathcal{L}_k} m^t_z \!\!\! \sum_{\tilde{a} \in \Sigma \times \{0,2\}} \!\!\! P^t_z(\tilde{a}) \cdot ( \log_2(P^t_z(\tilde{a})/P_z(\tilde{a})) + \log_2(1/P^t_z(\tilde{a}))) \\
&  \stackrel{\text{\eqref{def:D}}}{=} &  \sum_{z \in \mathcal{L}_k} m^t_z  \cdot (D(P^t_z \KL P_z) + H(P^t_z)) \\
& \stackrel{\text{\eqref{def:H_k}}}{\ge} &  H_k(t) 
\end{eqnarray*}
with equality in the last line if and only if $P^t_z = P_z$ for all $z \in \mathcal{L}_k$ with $m_z^t > 0$.
\end{proof}

\section{Tree straight-line programs and compression of binary trees} \label{sec-TSLP}

 We now introduce tree straight-line programs  and use them for the compression of binary trees.
 
\subsection{General tree straight-line programs}

 Let $V$ be a finite alphabet of symbols, where each symbol $A\in V$ has an associated rank $0$ or $1$ (we also speak of a ranked alphabet).
 The elements of $V$ are called \emph{nonterminals}.
 We assume that $V$ contains at least one nonterminal of rank $0$ and that $V$ is disjoint from the set $\Sigma \cup \{x\}$, which are the labels used for binary trees and contexts.
 We use $V_0$ (resp., $V_1$) for the set of nonterminals of rank $0$ (resp., of rank $1$).
 The idea is that nonterminals from $V_0$ (resp., $V_1$) derive to trees from $\T$ (resp., 
 contexts from $\cT$).
 We denote by $\T_V(\Sigma)$ the set of trees over $\Sigma \cup V$,
i.e., each node in a tree $t\in\T_V(\Sigma)$ is labeled with a symbol from $\Sigma \cup V$ such that nodes labeled by symbols from $\Sigma$ have zero or two children and
if a node is labeled by a symbol from $V$, then the number of children of this node corresponds to the rank of its label (a formal definition follows). 
With $\cT_V(\Sigma)$ we denote the corresponding set of all contexts, i.e., the set of trees
over  $\Sigma \cup \{x\} \cup V$, where the parameter symbol $x$ occurs exactly once and at a leaf position.
Formally, we define $\T_V(\Sigma)$ and $\cT_V(\Sigma)$
as the smallest sets of formal expressions with the following  conditions, where here and in the rest of the paper we use
the abbreviations $\T_{V}$ for $\T_V(\Sigma)$ and $\cT_V$ for $\cT_V(\Sigma)$:
\begin{itemize}
\item $\Sigma \cup V_0 \subseteq \T_V$ and $x \in \cT_V$,
\item if $a \in \Sigma$, $A\in V_1$ and $t_1,t_2\in\T_V$ then $A(t_1), a(t_1,t_2)\in\T_V$, and
\item if $a \in \Sigma$, $A\in V_1$, $s \in\cT_V$ and $t \in \T_V$ then $A(s), a(s,t), a(t,s) \in\cT_V$.
\end{itemize}
If e.g. $\Sigma=\{a,b\}$, $V_0=\{A\}$ and $V_1=\{B\}$, then $B(a(b(A,b),B(a)))\in\T_V$ and $B(a(b(A,b),B(x)))\in\cT_V$ as depicted in Figure~\ref{fig:treeswithnonterminals}.
Note that $\T(\Sigma) \subseteq \T_V(\Sigma)$ and $\cT(\Sigma) \subseteq \cT_V(\Sigma)$. 

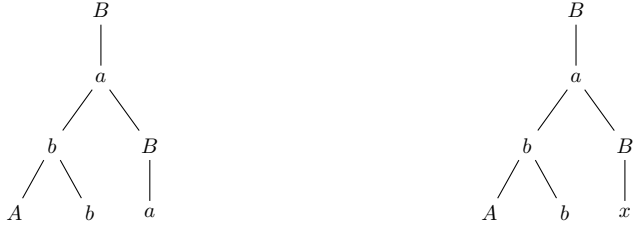
\begin{figure}
\begin{minipage}[hbt]{0.49\textwidth} 
\centering
\begin{tikzpicture}[-,level distance=9mm]
\tikzset{level 1/.style={sibling distance=23mm}}
\tikzset{level 2/.style={sibling distance=13mm}}
\tikzset{level 3/.style={sibling distance=10mm}}

\node [emptyKnot] (start){$B$}
  child {node [emptyKnot] (a) {$a$}
    child {node [emptyKnot] (c) {$b$}
      child {node [emptyKnot] (d) {$A$}}
      child {node [emptyKnot] (e) {$b$}}
    }
    child {node[emptyKnot]  (f) {$B$}
      child {node [emptyKnot] (g) {$a$}}
    }
  }
;
\end{tikzpicture}
\end{minipage}
\begin{minipage}[hbt]{0.49\textwidth} 
\centering
\begin{tikzpicture}[-,level distance=9mm]
\tikzset{level 1/.style={sibling distance=23mm}}
\tikzset{level 2/.style={sibling distance=13mm}}
\tikzset{level 3/.style={sibling distance=10mm}}

\node [emptyKnot] (start){$B$}
  child {node [emptyKnot] (a) {$a$}
    child {node [emptyKnot] (c) {$b$}
      child {node [emptyKnot] (d) {$A$}}
      child {node [emptyKnot] (e) {$b$}}
    }
    child {node[emptyKnot]  (f) {$B$}
      child {node [emptyKnot] (g) {$x$}}
    }
  }
;
\end{tikzpicture}\end{minipage}
\caption{Elements of $\T_V$ (left) and $\cT_V$ (right), where $a,b\in\Sigma$, $A\in V_0$ and $B\in V_1$.}
\label{fig:treeswithnonterminals}
\end{figure}

 A \emph{tree straight-line program} $\G$, or \emph{TSLP} for short, is a tuple $(V, A_0, r)$,
 where $A_0 \in V_0$ is the start nonterminal and 
$r:V\to (\T_V \cup \cT_V)$ is a function which assigns to each nonterminal its unique right-hand side.
It is required that if $A\in V_0$ (resp., $A \in V_1$), 
then $r(A) \in \T_V$ (resp., $r(A) \in \cT_V$).
Furthermore, the binary relation $\{ (A,B) \in V \times V \mid  B\text{ occurs in }r(A)\}$ has to be acyclic.
These conditions ensure that exactly one tree is derived from the start nonterminal $A_0$ 
 by using the rewrite rules $A \to r(A)$ for $A \in V$.
To define this formally, we define $\val_{\G}(t) \in \T$ for $t \in \T_V$ and 
$\val_{\G}(t) \in \cT$ for $t \in \cT_V$ inductively by the following rules:
\begin{itemize}
\item $\val_{\G}(a) = a$ for $a \in \Sigma$ and $\val_{\G}(x)=x$,
\item $\val_{\G}(a(t_1, t_2)) = a( \val_{\G}(t_1), \val_{\G}(t_2))$ for $a \in \Sigma$ and
$t_1, t_2 \in  \T_V \cup \cT_V$ (and $t_1 \in \T_V$ or $t_2 \in \T_V$ since there is at most one parameter in $a(t_1,t_2)$),
\item $\val_{\G}(A) = \val_{\G}(r(A))$ for $A \in V_0$,
\item $\val_{\G}(A(s)) = \val_{\G}(r(A)) [\val_{\G}(s) ]$ for $A \in V_1$ and $s \in  \T_V \cup \cT_V$ 
(note that $\val_{\G}(r(A))$ is a context $c$, so we can build $c[\val_{\G}(s) ]$).
\end{itemize}
 The tree defined by $\G$ is $\val(\G) = \val_{\G}(A_0)\in\T$. 
 
\begin{example}\label{example:TSLP}
	Let $\Sigma= \{a,b\}$ and 
	$\G = (\{A_0,A_1,A_2\}, A_0, r)$ be a TSLP such that $A_0,A_1 \in V_0, A_2 \in V_1$ and
	\[
	r(A_0)=a(A_1,A_2(b)),\; r(A_1) = A_2(A_2(b)), \; r(A_2) = b(x,a).
	\]
	We get $\val_{\G}(A_2) = b(x,a)$, $\val_{\G}(A_1) = b(b(b,a), a)$ and 
	$\val({\G}) = \val_{\G}(A_0) = a( b(b(b,a), a), b(b,a))$.
\end{example}

\subsection{Tree straight-line programs in normal form} \label{sec-normal-form}

In this section, we will use TSLPs in a certain normal form, which we introduce first.

A TSLP $\G = (V, A_0, r)$ is in {\em normal form} if the following conditions hold:
\begin{itemize}
\item $V = \{ A_0, A_1, \ldots, A_{m-1}\}$ for some $m \in \N$, $m \geq 1$.
\item For every $A_i \in V_0$, 
the right-hand side $r(A_i)$ is an expression of the form 
$A_j(\alpha)$, where $A_j \in V_1$ and $\alpha \in V_0 \cup \Sigma$.
\item For every $A_i \in V_1$ 
the right-hand side $r(A_i)$ is an expression of the form 
$A_j(A_k(x))$, $a(\alpha,x)$, or $a(x,\alpha)$,
where $A_j,A_k \in V_1$, $a \in \Sigma$ and $\alpha \in V_0 \cup \Sigma$.  
\item For every $A_i \in V$ define the word $\rho(A_i) \in (V \cup \Sigma)^*$ as follows:
$$
\rho(A_i) = \begin{cases}
A_j \alpha & \text{ if } r(A_i) = A_j(\alpha) \\
A_j A_k & \text{ if } r(A_i) = A_j(A_k(x)) \\
a\alpha & \text{ if } r(A_i) = a(\alpha,x) \text{ or } a(x,\alpha)
\end{cases}
$$
Let $\rho_{\G} = \rho(A_0) \rho(A_1) \cdots \rho(A_{m-1}) \in (\Sigma \cup \{A_1,A_2,\ldots, A_{m-1}\})^*$. 
Then we require that $\rho_{\G}$ is of the form 
$\rho_{\G} = A_1 u_1 A_2 u_2 \cdots A_{m-1}  u_{m-1}$ 
with $u_i \in (\Sigma \cup \{A_1,  A_2,\ldots, A_i \})^*$.
\item $\val_{\G}(A_i) \neq \val_{\G}(A_j)$ for $i \neq j$
\end{itemize}
We also allow the TSLP $\G_a = (\{A_0\}, A_0, A_0 \mapsto a)$ for every $a \in \Sigma$ in order to get the singleton tree $a$.
In this case, we set $\rho_{\G_a} = \rho(A_0) =  a$.

Let $\G = (V,A_0,r)$ be a TSLP in normal form with $V = \{ A_0, A_1, \ldots, A_{m-1}\}$ for the further definitions.
We define the size of $\G$ as $|\G| = |V| =m$.
Thus $2|\G|$ is the length of $\rho_{\G}$.
Let $\omega_{\G}$ be the word obtained from $\rho_{\G}$ by removing the first (i.e., left-most)
occurrence of $A_i$ from $\rho_{\G}$  for every $1\le i \le m-1$. Thus, if $\rho_{\G} = A_1 u_1 A_2 u_2 \cdots A_{m-1}  u_{m-1}$ 
with $u_i \in (\Sigma \cup \{ A_1,  A_2,\ldots, A_i \})^*$, then $\omega_{\G} = u_1 u_2 \cdots u_{m-1}$.
Note that $|\omega_{\G}| = |\rho_{\G}|-m+1 = m+1$.
The {\em entropy} $H(\G)$ of the normal form TSLP $\G$ is defined as the empirical unnormalized entropy of the word $\omega_{\G}$
(see \eqref{emp-entropy}):
$$
H(\G) = H(\omega_{\G}) .
$$
\begin{example}\label{example:TSLPnormalform}
Let $\Sigma= \{a,b\}$ and $\G=(\{A_0,A_1,A_2,A_3,A_4\},A_0,r)$ be the normal form TSLP with $A_0,A_2,A_3 \in V_0, A_1,A_4 \in V_1$ and
\begin{eqnarray*}
&&r(A_0)=A_1(A_2),\; r(A_1) = a(x,A_3), \; r(A_2) = A_4(A_3),\\
&&r(A_3)=A_4(b), \; r(A_4)=b(x,a).
\end{eqnarray*}
We have $\val(\G)=a( b(b(b,a), a), b(b,a))$, $\rho_{\G}=A_1A_2aA_3A_4A_3A_4bba$ ($u_1=u_3=\varepsilon$, $u_2 = a$, $u_4=A_3A_4bba$), $|\G|=5$ and $\omega_\G=aA_3A_4bba$. 
\end{example}

The {\em derivation tree} $T_{\G}$ of the normal form TSLP $\G$  is a binary tree with node labels from $V \cup \Sigma$.
The root is labeled with $A_0$. 
Nodes labeled with a symbol from $ \Sigma$ are the leaves of $T_{\G}$.
A node $v$ that is labeled with a nonterminal $A_i$ has $|\rho(A_i)|=2$ many children.
If $\rho(A_i) = \alpha \beta$ with $\alpha,\beta \in V \cup \Sigma$, then the left child of $v$ is labeled with $\alpha$ and the right child is labeled with $\beta$.
For every node $u$ of $T_{\G}$ we define the tree or context $s_u = \val_{\G}(\alpha)$
where $\alpha \in V \cup \Sigma$ is the label of $u$. If $\alpha \in V_0 \cup \Sigma$ then
$s_u \in \T$ and if $\alpha \in V_1$ then $s_u \in \cT$.
An {\em initial subtree} of the derivation tree $T_{\G}$ is a tree that can be obtained from $T_{\G}$ as follows:
Take a subset $U$ of the nodes of $T_{\G}$ and remove from $T_{\G}$ all proper descendants of nodes from $U$, i.e., all nodes that are located strictly below a node from $U$.

\begin{example} \label{example-derivation-tree}
Let $\G$ be the normal form TSLP from Example~\ref{example:TSLPnormalform}.
The derivation tree $T_{\G}$ is shown in Figure~\ref{fig-derivation-tree} on the left; an initial subtree $T'$ of it is shown on the right.
\end{example}

\begin{figure}[t]
		\tikzset{level 1/.style={sibling distance=24mm}}
		\tikzset{level 2/.style={sibling distance=12mm}}
		\tikzset{level 3/.style={sibling distance=6mm}}  
		\tikzset{level 4/.style={sibling distance=4mm}}
		\hspace*{\fill}
		\begin{tikzpicture}[scale=1,auto,swap,level distance=8mm]
		\node (eps) {$A_0$} 
		child {node {$A_1$}
		  child {node{$a$} }
			child {node {$A_3$}
				child {node {$A_4$}
					child {node{$b$}}
					child {node{$a$}}
				}
				child {node {$b$}}
			}
		}
		child {node {$A_2$}
			child {node {$A_4$}
			   child {node {$b$}}
				child {node {$a$}}
			}
			child {node {$A_3$}
				child {node {$A_4$}
				 child{node{$b$}}
					child {node{$a$}}
				}
				child {node {$b$}}
			}
		}
		;
		{label fig one};
		\end{tikzpicture}
		\hspace*{\fill}
		\begin{tikzpicture}[scale=1,auto,swap,level distance=8mm]
		\node (eps) {$A_0$} 
		child {node {$A_1$} 
		  child{node {$a$}}
			child {node {$A_3$}			 
				child {node {$A_4$}
				  child {node{$b$}}
					child {node{$a$}}
				}
				child {node {$b$}}
			}
		}
		child {node {$A_2$}
			child {node {$A_4$}
			}
			child {node {$A_3$}
			}
		}
		;
		{label fig one};
		\end{tikzpicture}
		\hspace*{\fill}
		\caption{The derivation tree $T_\G$ of the TSLP from Example~\ref{example-derivation-tree} (left) and an initial subtree $T'$ of $T_\G$ (right).}
		\label{fig-derivation-tree}
	\end{figure}
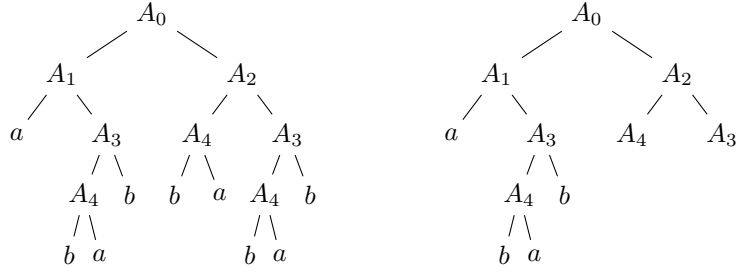

\begin{lemma} \label{lemma-number-leaves}
Let $\G$ be a TSLP in normal form with $t = \val(\G)$.
Let $T'$ be an initial subtree of $T_\G$ and let $v_1, \ldots, v_l$
be the sequence of all leaves of $T'$ (in left-to-right order). Then
$2|t| \geq \sum_{i=1}^l |s_{v_i}|$.
\end{lemma}

\begin{proof}
Let $u$ be a node of $T_\G$ and let $T_u$ be the subtree of $T_\G$ rooted in $u$.
Then, the nodes of $s_u$  are in a one-to-one correspondence with the leaves of $T_u$, that is, if $s_u \in \T$, we have $2|s_u|-1=|T_u|$ and if $s_u \in \cT$, we have $2|s_u| = |T_u|$ (recall that $|T_u|$ is the number of leaves of $T_u$).
Thus,  $2|s_u|-1 \leq |T_u|$. 
Since $T'$ is an initial subtree of $T_\G$ we get
$2|t|-1 = 2|\val(\G)|-1 = |T_\G| = \sum_{i=1}^l |T_{v_i}| \geq  \sum_{i=1}^l (2|s_{v_i}|- 1)$. Since $|s_{v_i}| \geq 1$ we get
$2|t| \geq \sum_{i=1}^l 2|s_{v_i}|-l +1 \geq \sum_{i=1}^l |s_{v_i}| +1$ and the statement follows.  
\end{proof}

A {\em grammar-based tree compressor} is an algorithm $\psi$ that produces for a given tree 
$t \in \T$ a TSLP $\mathcal{G}_t$ in normal form such that $t = \val(\G_t)$.
It is not hard to show that every TSLP can be transformed with a linear size increase into a normal form
TSLP that derives the same tree. For example, the TSLP from Example~\ref{example:TSLP} is transformed into the normal form TSLP described in Example~\ref{example:TSLPnormalform}.
We will not use this fact, since all we need is the following theorem from  \cite{GanardiHJLN17}
(recall that $\hat\sigma = \max\{2,\sigma\}$):

\begin{theorem} \label{thm-compression-ratio}
There exists a grammar-based compressor $\psi$ (working in linear  time) with 
$\max_{t \in \T_{n}} |\mathcal{G}_t| \leq  \mathcal{O}(n / \log_{\hat\sigma} n)$.
\end{theorem}

\subsection{Binary coding of TSLPs in normal form} \label{sec-binary-coding}

In this section we fix a binary encoding for normal form TSLPs. This encoding
is similar to the one for TSLPs producing unlabeled binary trees \cite{HuckeL17} (which in 
turn is based on the encoding for SLPs from \cite{KiYa00} and the encoding of DAGs from \cite{ZhangYK14}).
Let $\G = (V,A_0,r)$ be a TSLP in normal form with $m = |V| = |\G|$ nonterminals. 
We define the type $\type(A_i) \in \{0,1,2,3\}$ of a nonterminal $A_i \in V$ 
as follows: 
$$
\type(A_i) = \begin{cases}
0 & \text{ if } \rho(A_i) \in V_1 (V_0 \cup \Sigma) \\
1 & \text{ if } \rho(A_i) \in V_1 V_1 \\
2 & \text{ if } r(A_i) = a(\alpha,x) \text{ for some } \alpha \in V_0 \cup \Sigma \text{ and } a \in \Sigma\\
3 & \text{ if } r(A_i) = a(x,\alpha) \text{ for some } \alpha \in V_0 \cup \Sigma \text{ and } a \in \Sigma
\end{cases}
$$
We define the binary word $B(\G) = w_0 w_1 w_2 w_3 w_4$, where the words $w_i  \in \{0,1\}^+$, $0\le i \le 4$, are defined as follows:
\begin{itemize}
\item $w_0 = 0^{m-1}1$
\item $w_1 = a_0 b_0 a_1 b_1 \cdots a_{m-1} b_{m-1}$, where $a_j b_j$ 
is the 2-bit binary encoding of $\type(A_j)$. Note that $|w_1| = 2m$.
\item Let $\rho_{\G} = A_1 u_1 A_2 u_2 \cdots A_{m-1}  u_{m-1}$ with $u_i \in (\Sigma \cup \{ A_1,  A_2,\ldots, A_i \})^*$.
Then $w_2 = 1 0^{|u_1|} 1 0^{|u_2|} \cdots 1 0^{|u_{m-1}|}$. Note that $|w_2| = 2m$.
\item For $1\le i \le m-1$ let $k_i = |\rho_{\G}|_{A_i} \geq 1$ be the number of occurrences of the nonterminal $A_i$ in the word $\rho_{\G}$. Moreover, fix a total ordering on $\Sigma$. For $1 \leq i \leq \sigma$, let  $a_i$ denote the $i^{th}$ symbol in $\Sigma$ according to this ordering and let $l_i = |\rho_{\G}|_{a_i} \geq 0$ be the number of occurences of the symbol $a_i$ in the word $\rho_{\G}$. 
Then $w_3 = 0^{k_1-1} 1 0^{k_2-1} 1 \cdots 0^{k_{m-1}-1} 1 0^{l_1}1 0^{l_2}1 \cdots 0^{l_{\sigma}}1 $. 
Note that $|w_3| = 2m  + \sigma$.
\item The word $w_4$ encodes the word $\omega_{\G}$ using the well-known enumerative encoding \cite{Cover73}.
Every nonterminal $A_i$, $1\le i\le m-1$, has $\eta(A_i) := k_i-1$ occurrences in $\omega_{\G}$. Every symbol $a_i \in \Sigma$, $1 \leq i \leq \sigma$, has $\eta(a_i) = l_i$ occurences in $\omega_{\G}$. 
Let $S$ be the set of words over the alphabet $\Sigma \cup \{A_1,\ldots,A_{m-1}\}$ with $\eta(a_i)$ occurrences of $a_i \in \Sigma$ 
($1 \leq i \leq \sigma$) and 
$\eta(A_i)$ occurrences of $A_i$ ($1\le i \le m-1$). Hence, 
\begin{equation} \label{size-S}
|S| = \frac{(m+1)!}{ \prod_{i=1}^{\sigma}\eta(a_i)! \prod_{i=1}^{m-1} { \eta(A_i)!}} .
\end{equation}
Let $v_0, v_1, \ldots, v_{|S|-1}$ be the lexicographic enumeration of the words from $S$ with respect to 
the alphabet order $a_1, \dots, a_{\sigma}, A_1,\ldots,A_{m-1}$.
Then $w_4$ is the binary encoding of the unique index $i$ such that $\omega_{\G} = v_i$, where
$|w_4| = \lceil \log_2 |S| \rceil$ (leading zeros are added to the binary encoding of $i$ to obtain the length
$ \lceil \log_2 |S| \rceil$).
\end{itemize}

\begin{example}
Consider the normal from TSLP $\G$ from Example~\ref{example:TSLPnormalform}.
We have
$w_0=00001$,
$w_1 = 00 11 00 00 11$,
$w_2=1101100000$ and
$w_3=110101001001$.
To compute $w_4$, note first that there are $|S|=180$ words with two occurrences of $a$ and $b$ and one occurrence of $A_3$ and $A_4$.
It follows that $|w_4|=\lceil\log_2(180)\rceil=8$.
Furthermore, with the canonical ordering on $\Sigma=\{a,b\}$, the order of the alphabet is $a,b,A_3,A_4$. The word $\omega_{\G}=aA_3A_4bba$ is the lexicographically largest word in $S$ starting with $aA_3$. There are 132 words in $S$ that are  lexicographically larger than $aA_3A_4bba$, namely
all words in $S$ that start with $b$ (60 words), $A_3$ (30 words), $A_4$ (30 words), or $aA_4$ (12 words).
Hence $\omega_{\G}=aA_3A_4bba$ is the $48^{th}$ word in $S$ in lexicographic order, i.e., 
$\omega_G=v_{47}$ and thus $w_4=00101111$.
\end{example}

The following lemma generalizes a result  from \cite{HuckeL17}:

\begin{lemma}
The set of code words $B(\G)$, where $\G$ ranges over all TSLPs in normal form, is a prefix code.
\end{lemma}

\begin{proof}
Let $B(\G) = w_0 w_1 w_2 w_3 w_4$ with $w_i$ defined as above. We show how to recover the TSLP $\G$, given the alphabet $\Sigma$ and the ordering on $\Sigma$.
From $w_0$ we can determine $m = |V|$ and the factors $w_1$, $w_2$, and $w_3$ of 
$B(\G)$. Hence, we can determine the type of every nonterminal from $w_1$.
The types allow to compute $\G$ from the word $\rho_G$. Hence, it remains to determine $\rho_{\G}$.
To compute $\rho_\G$ from $w_2$, one only needs $\omega_{\G}$. 
For this, one determines the frequencies $\eta(A_1), \ldots, \eta(A_{m-1}), \eta(a_1), \dots, \eta(a_{\sigma})$ 
of the symbols in $\omega_{\G}$ from $w_3$. Using these frequencies one computes the size $|S|$ from \eqref{size-S}
and the length $\lceil \log_2 |S| \rceil$ of $w_4$. From $w_4$, one can finally compute $\omega_{\G}$.
\end{proof}

Note that $|B(\G)| \leq 7|\G|+\sigma + |w_4|$. By using the well-known bound on the code length 
of enumerative encoding \cite[Theorem~11.1.3]{CoTh06}, we get:
  
\begin{lemma} \label{lemma-binary-coding}
For the length of the binary coding $B(\G)$ we have
$$|B(\G)| \leq \mathcal{O}(|\G|) +\sigma+ H(\G).$$  
\end{lemma}

\section{Entropy bounds for binary encoded TSLPs} \label{sec-entropy-tslp}

For this section we fix a grammar-based tree compressor
$\psi: t \mapsto \G_t$ such that $\max_{t \in\T_n} |\G_t| \in \mathcal{O}(n / \log_{\hat\sigma} n)$; see Theorem~\ref{thm-compression-ratio}. Let $\gamma>0$ be a concrete constant such that 
\begin{equation} \label{eq-bound-G_t}
|\G_t| \le \frac{\gamma n}{\log_{\hat\sigma} n}
\end{equation}
for every tree $t \in\T_n$ and $n$ large enough.
We allow that the alphabet size $\sigma$ grows with $n$, i.e., $\sigma = \sigma(n)$ is a function in the tree size $n$
such that $1 \leq\sigma(n) \le 2n-1$ (a binary tree $t \in \T_n$ has $2n-1$ nodes).

We then consider the tree encoder $E_\psi : \T \to \{0,1\}^*$ defined
by $E_\psi(t) = B(\G_t)$. 

\begin{lemma} \label{lemma-entropy-bound}
Let $k \geq 0$, $t \in \T_n$ with $n \geq 2$ and let $\mathcal{P} = (P_w)_{w \in \mathcal{L}_k}$ be a $k^{th}$-order
tree process with $\Prob_{\mathcal{P}}(t) > 0$. We have
$$
H(\G_t) \leq  - \log_2 \Prob_{\mathcal{P}}(t) + \mathcal{O}\bigg( \frac{k n \log \hat\sigma}{\log_{\hat\sigma} n} \bigg) + \mathcal{O}\bigg(\frac{n \log\log_{\hat\sigma} n}{\log_{\hat\sigma} n} \bigg) .
$$
\end{lemma}

\begin{proof}
Let $m = |\G_t| = |V|$ be the size of $\G_t$.  
Let $T = T_{\G_t}$ be the derivation tree of $\G_t$. We define an initial subtree $T'$ as follows: If $v_1$ and $v_2$
are non-leaf nodes of $T$ that are labeled with the same nonterminal and $v_1$ comes before $v_2$
in preorder (depth-first left-to-right), then we remove from $T$ all proper descendants of $v_2$. Thus, for every $A_i \in V$ there 
is exactly one non-leaf node in $T'$ that is labeled with $A_i$.  
For the TSLP from Example~\ref{example:TSLPnormalform}, the tree $T'$ is shown in Figure~\ref{fig-derivation-tree} on the right.

Recall the definition of the words $\rho_{\G_t}$ and $\omega_{\G_t}$ from Section~\ref{sec-normal-form}. The word 
$\rho_{\G_t}$ can be obtained by writing down for every node $v$ of $T'$ the labels of $v$'s children and then concatenating these labels.
Moreover, the word $\omega_{\G_t}$ is obtained by writing down (in the right order) the labels of the leaves of $T'$.
Note that $T'$ has  $m$ non-leaf nodes and $m+1$ leaves.
Let $v_1, v_2, \ldots, v_{m+1}$ be the sequence of all leaves of $T'$ (w.l.o.g. in preorder)
and let $\alpha_i \in  \Sigma \cup \{A_1,\ldots,A_{m-1}\}$ be the label of $v_i$.
Let $\overline{\alpha} = (\alpha_1, \alpha_2, \ldots, \alpha_{m+1})$.
Then $\overline{\alpha}$ is a permutation of $\omega_{\G_t}$. We therefore have
$|\omega_{\G_t}|_\alpha = |\overline{\alpha}|_\alpha$ for every $\alpha \in \Sigma \cup \{A_1,\ldots,A_{m-1}\}$.
Hence, $p_{\overline{\alpha}}$ and $p_{\omega_{\G_t}}$ are the same empirical distributions. 
For the TSLP from Example~\ref{example:TSLPnormalform} we get 
$\overline{\alpha} = (a,b,a,b,A_4,A_3)$.
Let $s_i = \val_{\G_t}(\alpha_i) \in \T \cup (\cT \setminus \{x\})$.
Since $\val_{\G_t}(A_i) \neq \val_{\G_t}(A_j)$ for all $i \neq j$ ($\G_t$ is in normal form)
and $\val_{\G_t}(A_i) \notin \Sigma$ for all $i$ (this holds for every
normal form TSLP that produces a tree of size at least two), 
the tuple $\overline{s} = (s_1, s_2, \ldots, s_{m+1})$ satisfies for all $1\le i\le m+1$:
\begin{equation} \label{eq-equality-s_i}
p_{\omega_{\G_t}}(\alpha_i) = p_{\overline{s}}(s_i) .
\end{equation}
\begin{figure}
\centering
\begin{tikzpicture}[-,level distance=9mm]
\tikzset{level 1/.style={sibling distance=23mm}}
\tikzset{level 2/.style={sibling distance=13mm}}
\tikzset{level 3/.style={sibling distance=10mm}}

\node [emptyKnot, draw,circle,fill=gray,text opacity=1,opacity=0.15] (start){$a$}
  child {node [emptyKnot] (a) {$b$}
    child {node [emptyKnot] (c) {$b$}
      child {node[emptyKnot] (b1) {$b$}}
      child {node[emptyKnot] (h) {$a$}}
    }
    child {node[emptyKnot]  (f2) {$a$}}
  }
  child {node [emptyKnot, draw,circle,fill=gray,text opacity=1,opacity=0.15] (b) {$b$}
     child {node[emptyKnot, draw,circle,fill=gray,text opacity=1,opacity=0.15] (f3) {$b$}}
     child {node [emptyKnot, draw,circle,fill=gray,text opacity=1,opacity=0.15] (e) {$a$}}
  }
;
\draw[rounded corners, fill=gray,opacity=0.15]
($(c) + (0,.25)$) -- ($(h) + (.4,-.15)$) -- ($(b1) + (-.4,-.15)$) -- cycle;

\draw[rounded corners, fill=gray,opacity=0.15]
($(a) + (-0.25,0)$) -- ($(a) + (0,.25)$) -- ($(f2) + (0.25,0)$) -- ($(f2) + (0,-.25)$) -- cycle;
\end{tikzpicture}
\caption{The tree $\val(\G)$ of the TSLP from Example~\ref{example:TSLPnormalform}. The canonical occurrences of 
the trees/contexts or inner nodes in $(a, b, a, b, \val_{\G}(A_4),\val_{\G}(A_3)) = (a, b, a, b, b(x,a), b(b,a))$ 
used in the proof of Lemma~\ref{lemma-entropy-bound} are highlighted.}
\label{fig:subtrees}
\end{figure}
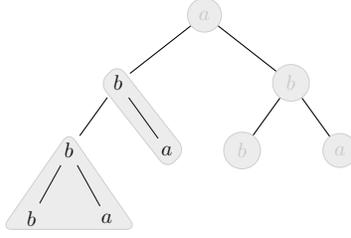
We define from  $\mathcal{P}$ for every $z \in \mathcal{L}_k$ a modified tree process
$\mathcal{P}_z = (P_{z,w})_{w \in \mathcal{L}}$ by setting 
\begin{equation} \label{def-P_z}
P_{z,w}(\tilde{a}) = P_{\ell_{k}(zw)}(\tilde{a})
\end{equation}
for all $\tilde{a} \in \Sigma \times \{0,2\}$. Note that 
the $k^{th}$-order tree process $\mathcal{P}$ is obtained for $z =  (\Box0)^k$ for the fixed padding symbol $\Box \in \Sigma$.
We define a mapping $\tau : \mathcal{T}\cup \mathcal{C} \rightarrow [0,1]$ by 
\begin{equation} \label{def-tau}
\tau(s) = \begin{cases}1 &\text{ if } s \in \mathcal{T}_1 = \Sigma \\
\max_{z \in \mathcal{L}_k} \Prob_{\mathcal{P}_z}(s) &\text{ if } s \in (\T \cup \cT) \setminus \T_1 
\end{cases}
\end{equation}
for every $s \in \T \cup \cT$.
Thus, for every $s \in (\T \cup \cT)\setminus \T_1$, the function $\tau$ maximizes the values of the function $\Prob_{\mathcal{P}}$ associated with the $k^{th}$-order tree process $\mathcal{P} = (P_w)_{w \in \mathcal{L}_k}$ by choosing an optimal $k$-history for the nodes of $s$ whose history is of length smaller than $2k$. 
We show that $\tau$ satisfies
\begin{equation}\label{submul}
\tau(t) \leq \prod_{i=1}^{m+1}\tau(s_i).
\end{equation}
In order to prove \eqref{submul}, first note that by definition of the tree/context $s_u$, for each node $u$ of the derivation tree $T$, the tree/context $s_u$ corresponds to a subtree/subcontext or a single inner node of the binary tree $t$. We define a function $\chi$ which maps a node $u$ of the derivation tree $T$ to a node 
$\chi(u) \in V(t) \subseteq \{0,1\}^*$: Intuitively, $\chi(u)$ is the root of the subtree/subcontext, respectively, the inner node of $t$ which corresponds to $s_u$. Formally, $\chi$ is defined inductively as follows: For the root node $u$ of $T$, we set $\chi(u) = \varepsilon$. Furthermore, let $u$ be a non-leaf node of $T$ which is labeled with the non-terminal $A_i$ and for which $\chi(u)$ has been defined. Let $u_1$ be the left child and $u_2$ be the right 
child of $u$ in $T$. We define $\chi(u_1) = \chi(u)$. The node $\chi(u_2)$ is defined as follows:
\begin{enumerate}[(i)]
\item  If $r(A_i)=A_j(\alpha)$ with $A_j \in V_1$ and $\alpha \in V \cup \Sigma$, then  
we set $\chi(u_2) = \chi(u) \omega(s_{u_1})$ (recall that $\omega(s_{u_1}) \neq \varepsilon$ is 
the position of the parameter $x$ in the context $s_{u_1} = \val_{\G}(A_j)$).
\item  If $r(A_i)=a(\alpha,x)$ (respectively, $r(A_i)=a(x, \alpha)$) for $a \in \Sigma$ and $\alpha \in \Sigma \cup V_0$, then we 
define $\chi(u_2)= \chi(u) 0$ (respectively, $\chi(u_2) = \chi(u)1$).
\end{enumerate}
This yields a well-defined function $\chi$ mapping a node $u$ of $T$ to a node $\chi(u) \in V(t)$. 
Let us define 
$$V_u = \{ \chi(u) v \mid v \in V(s_u) \} \subseteq V(t).$$
Then, the mapping 
\begin{equation} \label{eq-bijection}
V(s_u) \ni v \mapsto \chi(u)v \in V_u
\end{equation}
is bijective. The definition of the sets  $V_u$ implies 
that if two nodes $u$ and $v$ of $T$ are not in an ancestor-descendant relationship, then
$V_u \cap V_v = \emptyset$.
Since the nodes $v_1, \dots, v_{m+1}$ are the leaves of the initial subtree $T'$ and hence not in an ancestor-descendant relationship, the sets $V_i := V_{v_i}$ are disjoint subsets of $V(t)$. 
For the TSLP from Example~\ref{example:TSLPnormalform}, the node sets $V_1, V_2, V_3, V_4, V_5$ and $V_6$ corresponding to the six leaves of the initial subtree depicted in Figure~\ref{fig-derivation-tree} (right) are shown in Figure~\ref{fig:subtrees}.
Note that if $s_i \notin \mathcal{T}_1$ then the bijection from \eqref{eq-bijection} also preserves the 
$\lambda$-mapping in the following sense: 
\begin{equation} \label{eq:preserve-lambda}
 \lambda_t(\chi(v_i)w)=\lambda_{s_i}(w)
\end{equation} 
for every $w \in V(s_i)$. However, if
$s_i \in \mathcal{T}_1$ then this statement can be wrong since the number of children is not preserved in general: 
If $s_i \in \mathcal{T}_1$, then $s_i$ might correspond to a single inner node of $t$. In this case, we have $V_i = \{\chi(v_i)\}$, $V(s_i) = \{\varepsilon\}$ and $\lambda_t(\chi(v_i))=(a,2)$ for some $a \in \Sigma$, but $\lambda_{s_i}(\varepsilon)=(a,0)$. For example, in the TSLP from Example~\ref{example:TSLPnormalform}, the left-most leaf node of its initial subtree depicted in Figure~\ref{fig-derivation-tree}  corresponds to the root node of the tree $\val(\G)$ (see Figure~\ref{fig:subtrees}). We define 
$$\mathcal{I}:=\{i \in \{1, \dots, m+1\} \mid s_i \notin \T_1\}.$$
In our running example, we have
$(s_1, s_2, s_3, s_4, s_5, s_6) = (a, b, a, b, b(x,a), b(b,a))$ and hence $\mathcal{I}:= \{5,6\}$.

The history $h(\chi(v_i)w)$ of a node $\chi(v_i)w \in V_i$ with $w \in V(s_i)$ in the tree $t$ 
is the concatenation of the history $h(\chi(v_i))$ of $\chi(v_i)$ in $t$ and the history $h(w)$ of $w$ in the tree/context $s_i$.
Thus, if $i \in \mathcal{I}$, we have
\begin{align}
  \label{submul2}
  \begin{split}
  \max_{z \in \mathcal{L}_k}\prod_{v \in V_i}P_{\ell_{k}(zh(v))}(\lambda_t(v)) \ = \ & \max_{z \in \mathcal{L}_k}\prod_{w \in V(s_i)}P_{\ell_{k}(zh(\chi(v_i))h(w))}(\lambda_{t}(\chi(v_i)w))\\
   \  \stackrel{\text{\eqref{eq:preserve-lambda}}}{=} \  & \max_{z \in \mathcal{L}_k}\prod_{w \in V(s_i)}P_{\ell_{k}(zh(\chi(v_i))h(w))}(\lambda_{s_i}(w))\\
 \ \leq \ & \max_{z \in \mathcal{L}_k}\prod_{w \in V(s_i)}P_{\ell_{k}(zh(w))}(\lambda_{s_i}(w)).
\end{split}
\end{align}
For the inequality in the last line, note that every $k$-history $\ell_{k}(zh(\chi(v_i))h(w))$ for $z \in \mathcal{L}_k$ is also of the form
$\ell_{k}(z'h(w))$ for some $z' \in \mathcal{L}_k$.

We can now show \eqref{submul}. Since $t \in \T_n$ with $n \geq 2$ we have
\begin{eqnarray*} \allowdisplaybreaks
\tau(t) & \stackrel{\text{\eqref{def-tau}}}{=} & \max_{z \in \mathcal{L}_k}\Prob_{\mathcal{P}_z}(t) \\
& \stackrel{\text{\eqref{def-P_z}}}{=} & \max_{z \in \mathcal{L}_k}\prod_{v \in V(t)}P_{\ell_{k}(zh(v))}(\lambda_t(v)) \\
&\leq & \max_{z \in \mathcal{L}_k}\prod_{i \in \mathcal{I}}\prod_{v \in V_i}P_{\ell_{k}(zh(v))}(\lambda_t(v))   \quad \text{(since $P_{\ell_{k}(zh(v))}(\lambda(v)) \leq 1$ for $v \in V(t)$)}\\
&\leq & \prod_{i \in \mathcal{I}}\max_{z \in \mathcal{L}_k}\prod_{v \in V_i}P_{\ell_{k}(zh(v))}(\lambda_t(v)) \\
& \stackrel{\text{\eqref{submul2}}}{\le} & \prod_{i \in \mathcal{I}} \max_{z \in \mathcal{L}_k}\prod_{w \in V(s_i)}P_{\ell_{k}(zh(w))}(\lambda_{s_i}(w)) \\
&= & \prod_{i=1}^{m+1} \tau(s_i)  \quad 
\text{(since $\tau(s_i)=1$ for $i \notin \mathcal{I}$)} .
\end{eqnarray*}
Next, we define the function $\xi: \T \cup \cT \setminus \{x\} \rightarrow [0,1]$ as follows:
\begin{align*}
\xi(s) = \begin{cases}  
2^{-(k+2)}\sigma^{-(k+1)}\tau(s) \quad &\text{ if } s \in \T \\
\frac{6}{\pi^2}2^{-(k+1)}\sigma^{-k}\frac{\tau(s)}{|s|^2(|s|+1)} \quad &\text{ if } s \in \cT\setminus \{x\}.
\end{cases}
\end{align*}
We get
\begin{eqnarray*} \allowdisplaybreaks
\sum_{s \in \T \cup \cT \setminus \{x\}}\xi(s) 
&=& 2^{-(k+2)}\sigma^{-(k+1)}\sum_{s \in \T}\tau(s) + 
            \frac{6}{\pi^2}2^{-(k+1)}\sigma^{-k}\sum_{s \in \mathcal{C}\setminus \{x\}}\frac{\tau(s)}{|s|^2(|s|+1)}\\
&=& 2^{-(k+2)}\sigma^{-(k+1)}\bigg(\sum_{s \in \T \setminus \T_1} \max_{z \in \mathcal{L}_k} \Prob_{\mathcal{P}_z}(s) 
+ \sum_{s \in \T_1 }1 \bigg) +\\
& & \frac{6}{\pi^2}2^{-(k+1)}\sigma^{-k}\sum_{r \geq 1}\frac{1}{r^2(r+1)}\sum_{s \in \cT_r}
\max_{z \in \mathcal{L}_k} \Prob_{\mathcal{P}_z}(s)\\
&\leq & 2^{-(k+2)}\sigma^{-(k+1)}\bigg(\sum_{z \in \mathcal{L}_k}\sum_{s \in \mathcal{T}\setminus \mathcal{T}_1}  \Prob_{\mathcal{P}_z}(s) +\sigma \bigg)
+\\
& & \frac{6}{\pi^2}2^{-(k+1)}\sigma^{-k}\sum_{z \in \mathcal{L}_k}\sum_{r \geq 1}\frac{1}{r^2(r+1)}\sum_{s \in \mathcal{C}_r} \Prob_{\mathcal{P}_z}(s) \\
&\stackrel{\text{($\ast$)}}{\leq} & 2^{-(k+2)}\sigma^{-(k+1)}\left(2^{k}\sigma^k + \sigma \right)+ \frac{6}{\pi^2}2^{-(k+1)}\sigma^{-k} 2^{k}\sigma^k \sum_{r \geq 1}\frac{1}{r^2} \\
& \stackrel{\text{($\ast\ast$)}}{=} & 2^{-2}\sigma^{-1}+ 2^{-(k+2)} \sigma^{-k}+1/2 \\
& \leq  & 1,
\end{eqnarray*}
where ($\ast$) follows from Lemmas~\ref{lemma-sum-trees} and \ref{lemma-sum-contexts} and
$|\mathcal{L}_k| = 2^k\sigma^k$ and ($\ast\ast$) follows from the well-known fact that $\sum_{r \geq 1} r^{-2} = \pi^2/6$.
In particular, we have $\sum_{s \in \{s_1, \ldots, s_{m+1}\}} \xi(s) \leq 1.$
Thus, with Shannon's inequality \eqref{shannon}, we obtain:
\begin{equation*}
H(\mathcal{G}_t)=H(\omega_{\mathcal{G}_t})=\sum_{i=1}^{m+1} -\log_2 p_{\omega_{\G_t}}(\alpha_i) 
\stackrel{\text{\eqref{eq-equality-s_i}}}{=}
  \sum_{i=1}^{m+1}  -\log_2 p_{\overline{s}}(s_i) \leq  \sum_{i=1}^{m+1}  -\log_2 \xi(s_i) .
\end{equation*}
With $\mathcal{I}_0= \{i \mid 1 \le i \le m+1, s_i \in \mathcal{T}\}$ and  $\mathcal{I}_1 = \{ i \mid 1 \le i \le m+1, s_i \in \mathcal{C}\}$ we obtain
\begin{eqnarray*}
H(\mathcal{G}_t) &\leq& \sum_{i \in \mathcal{I}_0}-\log_2 \xi(s_i) + \sum_{i \in \mathcal{I}_1}-\log_2 \xi(s_i) \\
&=&\sum_{i \in \mathcal{I}_0}-\log_2\left( 2^{-(k+2)}\sigma^{-(k+1)}\tau(s_i) \right) + \\
&& \sum_{i \in \mathcal{I}_1}-\log_2\left(  \frac{6}{\pi^2}2^{-(k+1)}\sigma^{-k}\frac{\tau(s_i)}{|s_i|^2(|s_i|+1)} \right) 
\end{eqnarray*}
by definition of $\xi.$ Using logarithmic identities, we get
\begin{eqnarray*}
H(\mathcal{G}_t) &\leq & |\mathcal{I}_0|(k+2)
+ |\mathcal{I}_0|(k+1)\log_2\sigma
-\log_2\left(\prod_{i \in \mathcal{I}_0}\tau(s_i)\right) +\\
& &  \log_2\left(\frac{\pi^2}{6}\right)|\mathcal{I}_1|
+|\mathcal{I}_1|(k+1) +|\mathcal{I}_1|k\log_2\sigma -\log_2\left(\prod_{i \in \mathcal{I}_1}\tau(s_i)\right)
+ \\
& & \sum_{i \in \mathcal{I}_1}\log_2 |s_i|^2(|s_i|+1).
\end{eqnarray*}
Using $|\mathcal{I}_0|+|\mathcal{I}_1|=m+1 \leq 2m = 2|\mathcal{G}_t|$,  $\log_2(\pi^2/6)|\mathcal{I}_1|\leq |\mathcal{I}_1|$ and $|s_i|+1\leq 2|s_i|$, we obtain
\begin{equation*}
H(\mathcal{G}_t) \leq 2(k+2)|\mathcal{G}_t| 
+2(k+1)|\mathcal{G}_t|\log_2\sigma
-\log_2\left(\prod_{i=1}^{m+1}\tau(s_i)\right)+\sum_{i =1}^{m+1}\log_2 2|s_i|^3.
\end{equation*}
Equation~\eqref{submul} and  $\tau(t) \geq \Prob_{\mathcal{P}}(t)$ yield
\begin{eqnarray*}
H(\mathcal{G}_t)&\leq & 2(k+3)|\mathcal{G}_t| +2(k+1)|\mathcal{G}_t|\log_2\sigma -\log_2\tau(t) +3\sum_{i =1}^{m+1}\log_2|s_i| \\
&\le & -\log_2\Prob_{\mathcal{P}}(t)  + \mathcal{O}(k |\mathcal{G}_t| \log \hat\sigma + \sum_{i =1}^{m+1}\log_2|s_i|).
\end{eqnarray*}
Let us bound the sum $\sum_{i =1}^{m+1}\log_2|s_i|$: 
Using Jensen's inequality and Lemma~\ref{lemma-number-leaves}  (which yields
$\sum_{i=1}^{m+1} |s_i| \leq 2n$), we get
\begin{eqnarray*}
\sum_{i =1}^{m+1}\log_2|s_i| & \le & (m+1)\log_2\left(\sum_{i =1}^{m+1} \frac{|s_i|}{m+1} \right) \\
& \le & (m+1)\log_2\left( \frac{2n}{m+1} \right) \\
&\le & 2 |\mathcal{G}_t| \log_2\left( \frac{2n}{|\mathcal{G}_t|} \right)
\end{eqnarray*} and thus
\begin{equation} \label{entropy-bound-intermediate}
H(\mathcal{G}_t) \le -\log_2\Prob_{\mathcal{P}}(t)  + \mathcal{O}\bigg(k |\mathcal{G}_t| \log \hat\sigma + |\mathcal{G}_t| \log_2\left( \frac{n}{|\mathcal{G}_t|} \right)\bigg) .
\end{equation}
To bound the term $|\mathcal{G}_t| \log_2(n/|\mathcal{G}_t|)$ recall that for $n$ large enough we have
$|\mathcal{G}_t| \le \gamma \cdot n/\log_{\hat\sigma} n = \gamma \cdot n \cdot \log \hat\sigma / \log n$ 
by \eqref{eq-bound-G_t}. Here, $\gamma$ is a constant. Since 
$\sigma \le 2n-1$ there is a constant $\gamma' \geq 1$ with $\gamma \cdot n/\log_{\hat\sigma} n \leq \gamma' n$.
Since  for every fixed $z \geq 1$,
the function $\phi(x) = x\log_2\left(\frac{z}{x}\right)$ is monotonically increasing for $0 < x \le \frac{z}{e}$ (where $e$ is Euler's number), we get
\begin{align*}
|\mathcal{G}_t| \log_2\bigg(\frac{n}{|\mathcal{G}_t|}\bigg)  \le 
|\mathcal{G}_t| \log_2\bigg(\frac{e \gamma' n}{|\mathcal{G}_t|}\bigg) \le
\frac{\gamma n \log_2\big(\frac{e \gamma'}{\gamma} \log_{\hat\sigma} n\big)}{\log_{\hat\sigma} n} 
\le
\mathcal{O}\bigg(\frac{n \log\log_{\hat\sigma} n}{\log_{\hat\sigma} n}\bigg).
\end{align*}
With \eqref{entropy-bound-intermediate} we get
\begin{align*}
H(\G_t) \leq  - \log_2 \Prob_{\mathcal{P}}(t) + \mathcal{O}\bigg( \frac{k n \log \hat\sigma}{\log_{\hat\sigma} n} \bigg) + \mathcal{O}\bigg(\frac{n \log\log_{\hat\sigma} n}{\log_{\hat\sigma} n} \bigg),
\end{align*}
which proves the lemma.
\end{proof}

\begin{theorem} \label{thm-red-bound}
For every $t \in \T_n$ and every $k \geq 0$ we have
$$
|E_\psi(t)| \le H_k(t) + \mathcal{O}\bigg( \frac{k n \log \hat\sigma}{\log_{\hat\sigma} n} \bigg) + \mathcal{O}\bigg(\frac{n \log\log_{\hat\sigma} n}{\log_{\hat\sigma} n} \bigg)+\sigma .
$$
\end{theorem}

\begin{proof}
Let $\mathcal{P} = (P_w)_{w \in \{0,1\}^k}$ be a $k^{th}$-order
tree process with $P(t) > 0$.
Lemmas ~\ref{lemma-binary-coding} and \ref{lemma-entropy-bound} yield
\begin{eqnarray*}
|E_\psi(t)| &\leq& \mathcal{O}(|\G_t|) + H(\G_t) +\sigma\\
& \leq & \mathcal{O}(|\G_t|) - \log_2 \Prob_{\mathcal{P}}(t) + \mathcal{O}\bigg( \frac{k n \log \hat\sigma}{\log_{\hat\sigma} n} \bigg) + \mathcal{O}\bigg(\frac{n \log\log_{\hat\sigma} n}{\log_{\hat\sigma} n} \bigg) +\sigma\\
& = & - \log_2 \Prob_{\mathcal{P}}(t) + \mathcal{O}\bigg( \frac{k n \log \hat\sigma}{\log_{\hat\sigma} n} \bigg) + \mathcal{O}\bigg(\frac{n \log\log_{\hat\sigma} n}{\log_{\hat\sigma} n} \bigg)+\sigma,
\end{eqnarray*}
where the last equality uses the bound $|\G_t| \in \mathcal{O}(n / \log_{\hat\sigma} n)$.
Finally, by taking for $\mathcal{P}$ be the empirical $k^{th}$-order tree process  $\mathcal{P}^t$,
we get 
$$
|E_\psi(t)| \leq H_k(t) + \mathcal{O}\bigg( \frac{k n \log \hat\sigma}{\log_{\hat\sigma} n} \bigg) + \mathcal{O}\bigg(\frac{n \log\log_{\hat\sigma} n}{\log_{\hat\sigma} n} \bigg)
+\sigma
$$
from Theorem~\ref{theo-travis-for-trees}.
\end{proof}

\section{Extension to unranked trees}\label{extension-unranked}
So far, we have only considered binary trees.
In this section, we consider unranked, ordered trees, where the number
of  children of a node (also called its degree) can be any natural number and the children of every node are totally ordered. As before, each node is labeled by an element of some finite alphabet $\Sigma$.
Let us denote by $\mathcal{U}(\Sigma)$ (or simply $\mathcal{U}$) the set of all such trees. For technical reasons
we also define {\em forests} which are ordered sequences of trees from $\mathcal{U}$.
The set of forests is denoted with $\mathcal{F}$. The sets $\mathcal{U}$
and $\mathcal{F}$ can be inductively defined as the smallest sets of strings over the alphabet $\Sigma \cup \{ (, )\}$
such that the following conditions hold:
\begin{itemize}
\item $\varepsilon \in \mathcal{F}$ (this is the empty forest),
\item if $a \in \Sigma$ and $f \in \mathcal{F}$ then $a(f) \in \mathcal{U}$,
\item if $t \in \mathcal{U}$ and $f \in \mathcal{F}$ then $tf \in \mathcal{F}$.
\end{itemize}
The singleton tree $a()$ (which is obtained by taking $f = \varepsilon$ in the second point)
is usually written as $a$. Note that $\mathcal{U} \subseteq \mathcal{F}$ and that $\mathcal{F} = \mathcal{U}^*$. 
The size $|f|$ of $f \in \mathcal{F}$ is the number of occurrences of $\Sigma$-labels in $f$; formally: 
$|\varepsilon| = 0$, $|a(f)| = 1 + |f|$ and $|tf| = |t|+|f|$ for $a \in \Sigma$, $t \in \mathcal{U}$, and $f \in \mathcal{F}$.

The {\em first-child/next-sibling encoding} transforms a forest $f \in \mathcal{F}$
 into a binary tree $\ffcns(f) \in \T$. It is defined inductively as follows (recall that $\Box \in \Sigma$ is 
 a fixed distinguished symbol in $\Sigma$):
 \begin{itemize}
\item $\ffcns(\varepsilon) = \Box$ and
\item $\ffcns(a(f)g)  = a(\ffcns(f), \ffcns(g))$ for
$f, g \in \mathcal{F}$ and $a \in \Sigma$.
\end{itemize}
Thus, the left (resp., right) child of a node in $\ffcns(f)$ is the first child (resp., right sibling)
of the node in $f$ or a $\Box$-labeled leaf if it does not exist.
\begin{example}
\label{example:fcns}
If $f = a(bc)d(e)$ then
\begin{equation*}
\begin{split}
\ffcns(f) &= \ffcns(a(bc)d(e))  =  a(\ffcns(bc), \ffcns(d(e)) \\
&= a(b(\Box,\ffcns(c)), d(\ffcns(e), \Box)) =  a(b(\Box, c(\Box,\Box)),  d(e(\Box,\Box),\Box)) ,
\end{split}
\end{equation*}
see also Figure~\ref{fig-fcns}.

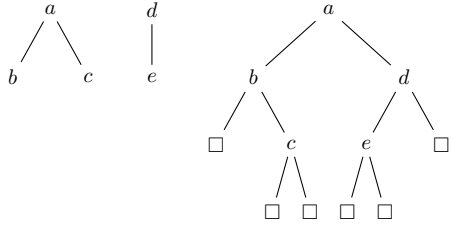
\begin{figure}
\centering
\begin{tikzpicture}[-,level distance=9mm]
\tikzset{level 1/.style={sibling distance=10mm}}
\tikzset{level 2/.style={sibling distance=8mm}}

\node [emptyKnot] (a){$a$}
  child {node [emptyKnot] (b) {$b$}}
  child {node [emptyKnot] (c) {$c$}} ;

\node [emptyKnot, right = 1cm of a] (d){$d$}
  child {node [emptyKnot] (e) {$e$}} ;
  
\tikzset{level 1/.style={sibling distance=20mm}}
\tikzset{level 2/.style={sibling distance=10mm}}
\tikzset{level 3/.style={sibling distance=5mm}}
  
\node [emptyKnot, right = 2cm of d] (a'){$a$}
  child {node [emptyKnot] (b') {$b$}
      child {node [emptyKnot] (box1) {$\Box$}}
      child {node [emptyKnot] (c') {$c$}
        child {node [emptyKnot] (box2) {$\Box$}}
        child {node [emptyKnot] (box3) {$\Box$}}}}
  child {node [emptyKnot] (d') {$d$}
   child {node [emptyKnot] (e') {$e$}
        child {node [emptyKnot] (box5) {$\Box$}}
        child {node [emptyKnot] (box6) {$\Box$}}}
   child {node [emptyKnot] (box4) {$\Box$}}}
   ;
\end{tikzpicture}
\caption{\label{fig-fcns}Forest $f$ on the left and $\ffcns(f)$ on the right from Example~\ref{example:fcns}.}
\end{figure}
\end{example}
Note that if $t \in \mathcal{U}$, $|t| = n$ then $\ffcns(t)$ is a binary tree with $n$ internal nodes. Hence
we have $|\ffcns(t)| = n+1$ (which is the number of leaves of $\ffcns(t)$). 
We define the $k^{th}$-order empirical entropy of an unranked tree $t \in \mathcal{U}$ as 
$H_k(t)=H_k(\ffcns(t))$. Note that this definition is independent of the choice of the symbol $\Box \in \Sigma$. From Theorem \ref{thm-red-bound}, we immediately obtain:
\begin{theorem}
For every $t \in \mathcal{U}$ with $|t| = n$ and every $k \geq 0$ we have
$$
|E_\psi(\ffcns(t))| \le H_k(t) + \mathcal{O}\bigg( \frac{k n \log \hat\sigma}{\log_{\hat\sigma} n} \bigg) + \mathcal{O}\bigg(\frac{n \log\log_{\hat\sigma} n}{\log_{\hat\sigma} n} \bigg)+\sigma .
$$
\end{theorem}
The above definition of the $k^{th}$-order empirical entropy of an unranked tree can be also applied
to binary trees $t$ (a binary tree can be viewed as a particular unranked tree).
This yields $H_k(\ffcns(t))$ and leads to the question how this value relates to $H_k(t)$ (the 
$k^{th}$-order empirical entropy of $t$ as defined before in \eqref{def:H_k}).
In one direction, we have the following bound:

\begin{lemma} \label{lemma-H_2k}
Let $t \in \mathcal{T}(\Sigma)$ denote a binary tree with first-child next-sibling encoding $\ffcns(t) \in \mathcal{T}(\Sigma)$. Then $H_{2k}(\ffcns(t))\leq H_{k-1}(t)$ for $1 \leq k \leq |t|$.  
\end{lemma}
The somewhat technical proof of Lemma~\ref{lemma-H_2k} can be found in Appendix~\ref{appendix-H_2k}.
In contrast to Lemma~\ref{lemma-H_2k}, there are families of binary trees $t_n$ where $H_k(\ffcns(t_n))$ is exponentially smaller than $H_k(t_n)$ for every $n \geq 1$ and $k \geq 2$: 
Define $t_n$ inductively by $t_1 = a$ and $t_n = a(c, t_{n-1})$ if $n$ is even and $t_n = a(b,t_{n-1})$ if $n$ is odd. Thus, $t_n$ denotes a right-degenerate binary tree of size $n$, whose inner nodes and right-most leaf are labeled with $a$ and whose leaves except for the right-most leaf are alternately labeled $b$ and $c$. We get 
$H_k(t_n) \in \Theta(n-k)$: there are $n-k$ many nodes $v$ with $k$-history $(a1)^{k-1}a0$, and about half of them are $b$-labeled leaves, while the other half are
$c$-labeled leaves.
Moreover, we have $H_k(\ffcns(t_n)) \in \Theta(\log(n-k))$: the fcns-encodings of the binary trees $t_n$ can be inductively defined by $\ffcns(t_1) = a(\Box, \Box)$ and $\ffcns(t_n) = a(c(\Box,\ffcns(t_{n-1})),\Box)$ if $n$ is even and $\ffcns(t_n) = a(b(\Box,\ffcns(t_{n-1})),\Box)$ if $n$ is odd. Intuitively, as the labels $b$ and $c$ are thus incorporated in $k$-histories of nodes of $\ffcns(t_ n)$, we can thus determine the label of a node from its $k$-history for $k \geq 2$ for most nodes of $\ffcns(t_n)$.

Our definition of the $k^{th}$-order empirical entropy of an unranked tree via the fcns-encoding
has a practical motivation. Unranked trees occur for instance in the context of XML, where the hierarchical
structure of a document is represented as an unranked node labeled tree. In this setting, the label of a node
quite often depends on (i) the labels of the ancestor nodes and (ii) the labels of the (left) siblings. This dependence 
is captured by our definition of the  $k^{th}$-order empirical entropy.

We also confirmed this intuition by experimental data (shown in Table~\ref{table}) with real XML document trees
(ignoring textual data at the leaves) showing that in these cases the  
$k^{th}$-order empirical entropy is indeed very small compared to the worst-case bit size.
More precisely, we computed for 21 real XML document trees\footnote{All data are available from \url{http://xmlcompbench.sourceforge.net/Dataset.html}.}
the $k^{th}$-order empirical entropy (for $k=1,2,4,8$) and divided the value by the worst-case bit length
$2n + \log_2(\sigma) n$, where $n$ is the number of nodes and $\sigma$ is the number of node labels \cite{GearyRR06}.

Our experimental results combined with our entropy bound~\eqref{entropy-bound} for grammar-based
compression are in accordance with the fact that grammar-based tree compressors yield impressive compression ratios
for XML document trees, see e.g. \cite{LohreyMM13}. Some of the XML documents from our experiments were 
also used in \cite{LohreyMM13}, where the performance of the grammar-based tree compressor TreeRePair was tested.
An interesting observation is that those XML trees, for which 
our $k$-th order empirical entropy is large are indeed those XML trees with the worst
compression ratio for TreeRePair in \cite{LohreyMM13}. This is in particular true for the Treebank document, see 
Table~\ref{table}. TreeRePair obtained for Treebank a compression ratio of around 20\%, whereas for all other documents
tested in \cite{LohreyMM13} TreeRePair achieved a compression ratio below 8\%.

\begin{table}[h!]
\centering\scriptsize
    \begin{tabular}{lrrrrrrr}
    XML document            & $n$ & $\sigma$  & $w := (2 +\log_2\sigma)n$  & $H_1/w$  & $H_2/w$ & $H_4/w$  & $H_8/w$   \\
    \midrule
    \fname{Baseball}      &28\ 306&46&212\ 961.9447&2.9818 \%&1.2547 \%&0.6739 \%&0.6662 \% \\
    \fname{DBLP}           &3\ 332\ 130&35&23\ 755\ 697.8193&10.9775 \%&8.7407 \%&8.2134 \%&6.7270 \% \\
    \fname{DCSD-Normal}&2\ 242\ 699&50&17\ 142\ 868.6330&4.2437 \%&2.2481 \%&1.7517 \%&1.3038 \% \\
    \fname{EnWikiNew}         &404\ 652&20&2\ 558\ 180.8475&9.5317 \%&3.0760 \%&3.0759 \%&2.9378 \%\\
    \fname{EnWikiQuote}      &262\ 955&20&1\ 662\ 382.6021&9.4270 \%&3.1014 \%&3.1014 \%&3.1006 \%\\  
    \fname{EnWikiVersity}      &495\ 839&20&3\ 134\ 658.5046&8.8952 \%&2.3753 \%&2.3753 \%&2.3750 \% \\
      \fname{EXI-Array}           &226\ 523&47&1\ 711\ 288.1304&0.2506 \%&0.2495 \%&0.2492 \%  &0.2483 \%  \\            
      \fname{EXI-factbook}      &55\ 453&199&534\ 379.7451&2.2034 \%&0.9450 \%&0.8132 \%&0.8092 \% \\
     \fname{EXI-Invoice}        &15\ 075&52  &116\ 084.1288&0.0484 \%&0.0268 \%&0.0139 \%&0.0098 \% \\
    \fname{EXI-Telecomp}&177\ 634&39 & 1\ 294\ 135.1377&1.5405 \% &0.0044 \% &0.0034 \%&0.0021 \%  \\   
    \fname{EXI-weblog}        &93\ 435 &12 &521\ 830.9713&0.0032 \%&0.0028 \% &0.0028 \%&0.0028 \% \\    
    \fname{Lineitem}         &1\ 022\ 976&18&6\ 311\ 685.1983&0.0003 \%&0.0003 \%&0.0003 \%&0.0003 \%\\
    \fname{Mondial}           &22\ 423&23&146\ 277.8297&11.1285 \%&9.2940 \%&8.4702 \%&7.7679 \%\\
    \fname{NASA}               &476\ 646&61&3\ 780\ 154.2290&7.7424 \%&4.4588 \%&3.8898 \%&3.8054 \%\\   
    \fname{Shakespeare}       &179\ 690&22&1\ 160\ 695.2676&11.9140 \%&10.8416 \%&10.6368 \%&10.4765 \%\\
    \fname{SwissProt}   &2\ 977\ 031&85&25\ 035\ 017.5080&12.1892 \%&10.5249 \%&9.2455 \%&8.1204 \%\\   
    \fname{TCSD-Normal}   &2\ 749\ 751&24&18\ 107\ 007.2213&8.5450 \%&8.4004 \%&8.2862 \%&8.2472 \% \\   
    \fname{Treebank}     &2\ 437\ 666&250&24\ 293\ 253.5140 &30.8912 \%&23.0825 \%&19.2444 \%&13.4058 \%\\
    \fname{USHouse}     &6\ 712&43&49\ 845.0890&21.0500 \%&18.2164 \%&12.6572 \%&9.3754 \%\\   
     \fname{XMark1}          &167\ 865&74&1\ 378\ 079.8892&12.1610 \%&9.5101 \%&9.2271 \%&8.4281 \% \\  
   \fname{XMark2}         &1\ 666\ 315&74&13\ 679\ 535.2849&12.2125 \%&9.5634 \%&9.3259 \%&8.9400 \% \\   \\
    \end{tabular}
  \caption{Experimental results for XML tree structures, where $n$ denotes the number of nodes and $\sigma$ denotes
  the number of node labels.} \label{table}
\end{table}

\section{String straight-line programs versus higher-order empirical entropy of strings} \label{sec-string-SLP}

Our definition of $k^{th}$-order empirical entropy does not capture all regularities that can be exploited in grammar-based
compression. Take for instance a complete unlabeled binary tree $t_n$ of height $n$ (all paths from the root to a leaf have length $n$).
This tree has $2^n$ leaves and is very well compressible: its minimal DAG has only $n+1$ nodes, hence there also exists
a TSLP of size $n+1$ for $t_n$. But for every fixed $k$ the $k^{th}$-order empirical entropy of $t_n$ divided by $n$ converges to $2$ (the trivial
upper bound) for $n \to \infty$. If $n \gg k$ then for every $k$-history $z$ the number of leaves with $k$-history 
$z$ is roughly the same as the number of internal nodes with $k$-history $z$. Hence, although $t_n$ is highly compressible with TSLPs
(and even DAGs), its $k^{th}$-order empirical entropy is close to the maximal value. 
We show in the following that the same phenomenon occurs for grammar-based string compression and the well-established empirical entropy of strings.

The $k^{th}$-order empirical entropy of a string is defined as follows (see e.g.~\cite{Gagie06a}). Let $\Sigma$ denote a finite alphabet and let $w \in \Sigma^*$. 
For a non-empty string $\alpha \in \Sigma^+,$ define $w(\alpha) \in \Sigma^*$ as the string whose $i^{th}$ symbol is the symbol in $w$ immediately following the $i^{th}$ occurrence of the string $\alpha$ in $w$. Thus,  if $\alpha$ is not a suffix of $w$, the length of $w(\alpha)$ is equal to the number of occurrences of the string $\alpha$ in $w$. In case $\alpha$ is a suffix of $w$, $|w(\alpha)|$ is the number of occurrences of $\alpha$ in $w$ minus one. 
Recall the definition of the unnormalized empirical entropy $H(w)$ of a string $w \in \Sigma^+$  (or tuple) 
from Section~\ref{sec-empirical}.
For an integer $k \geq 1$, the \emph{$k^{th}$-order (unnormalized) empirical entropy} of a string $w \in \Sigma^+$ is defined as
\begin{align*}
H_k(w) = \sum_{ \alpha \in \Sigma^k}H(w(\alpha)),
\end{align*}
where we set $H(\varepsilon) = 0$. For $k = 0$, $H_0(w) = H(w)$ is the  (unnormalized) empirical entropy of $w$.

A \emph{straight-line program} (SLP) for a string $w$ is a context-free grammar that produces only the string $w$. The size of an SLP is the sum of the lengths of the right-hand sides of the production rules of the context-free grammar, see e.g. \cite{Loh12survey} for details.
We prove that for each $n \geq 1$  there exists a string of length $2^{n+1}-1$, which is highly compressible with SLPs, 
but whose $k^{th}$-order empirical entropy is close to the maximum.

\begin{theorem}\label{theo-string-vergleich}
There exists a family of strings $(S_n)_n$ ($n \geq 1$) over a binary alphabet with the following properties:
\begin{itemize}
\item $|S_n| = 2^{n+1}-1$,
\item there exists an SLP of size $3n$ for $S_n$, and 
\item $H_k(S_n) \geq 2^{n+1-k}(1-o(1))$ for $k \in o(n)$.
\end{itemize}
\end{theorem}

\begin{proof}
We inductively define a string $S_n \in \{a,b\}^*$ for $n \geq 1$ as follows: We set 
\begin{itemize}
\item $S_1 = baa$ and 
\item $S_n=bS_{n-1}S_{n-1}$. 
\end{itemize}
We  have $|S_n| = 2^{n+1}-1$. The string $S_n$ corresponds to the preorder traversal of the perfect binary tree of size $2^n$, whose internal nodes are labeled with the symbol $b$ and whose leaves are labeled with the symbol $a$. 
The recursive definition of $S_n$ directly translates to an SLP for $S_n$ of size $3n$ (there is a nonterminal for each $S_i$ with $1\le i\le n$ and each rule has three symbols on the right-hand side according to the recursive definition).

It remains to show that $H_k(S_n) \geq 2^{n-k}$ for 
$0 \leq k < n$. We start with the case $k=0$. Recall that $|w|_x$ denotes the number of occurrences of a symbol $x$ in a string $w$, as defined in 
Section~\ref{sec-prelim}.
We have $|S_n|_a = 2^n$ and $|S_n|_b = 2^n-1$, which yields
\begin{eqnarray*}
H(S_n)&=&\sum_{x \in \{a,b\}} |S_n|_x \log_2\left( \frac{|S_n|}{|S_n|_x}\right) \\
&=& 2^n \log_2\left(\frac{2^{n+1}-1}{2^n}\right)
+ (2^{n}-1)\log_2\left(\frac{2^{n+1}-1}{2^n-1}\right).
\end{eqnarray*}
Define the function $g : [2,\infty) \to \mathbb{R}$ by
$$
g(x) = \frac{x}{2x-1} \log_2\left(\frac{2x-1}{x}\right) +\frac{x-1}{2x-1}\log_2\left(\frac{2x-1}{x-1}\right).
$$ 
It converges to $1$ from below for $x \to \infty$. 
Since $|S_n| = 2^{n+1}-1$ we have $H(S_n) = g(2^n) |S_n| \geq 2^{n+1}(1-o(1))$.

Let us now consider the case  $k \geq 1$ 
 and let $1 \leq m \leq n$. By construction of $S_n$, the last symbol of $S_n$ is $a$.
Therefore, the length of the string $S_n(b^m)$ equals the number of occurrences of the string $b^m$ in $S_n$. In order to lower-bound the $k^{th}$-order empirical entropy of $S_n$, we first show inductively in $n$, that 
\begin{equation} \label{eq:S_n(b^m)}
 |S_n(b^m)| = 2^{n-m+1}-1
\end{equation}
 for $1 \leq m \leq n$: For the base case, let $n = 1$. We have $S_1=baa$ and thus, $|S_1(b)|=1$. For the induction step, let $n > 1$. By definition of $S_n$, we have $S_n = bS_{n-1}S_{n-1}$. 
By the induction hypothesis, we have $|S_{n-1}(b^m)| = 2^{n-m}-1$ for $1 \leq m \leq n-1$. Moreover, $b^n$ does not occur
in $S_{n-1}$ (which follows by induction), i.e., $|S_{n-1}(b^n)| = 0 = 2^{n-n}-1$.
By construction, the last symbol of the string $S_{n-1}$ is $a$. Thus, for all $1 \leq m \leq n$ we have
$|S_{n-1}S_{n-1}(b^m)| = 2 |S_{n-1}(b^m)| = 2^{n-m+1}-2$. Hence, as the string $b^m$ with $1 \leq m \leq n$ occurs additionally as a prefix of the string $S_n=bS_{n-1}S_{n-1}$, the number of occurrences of $b^m$ in $S_n$ in total is $|S_n(b^m)| = 2^{n-m+1}-1$ for every $1 \leq m \leq n$. This proves 
\eqref{eq:S_n(b^m)}. 

Next, we count the number of occurrences of $b^m$ in $S_n$, which are followed by the symbol $a$, that is, we count $|S_n(b^m)|_a$. We show inductively in $n$, that 
$$|S_n(b^m)|_a=2^{n-m}$$
for $1 \leq m \leq n$:
For the base case, let $n=1$. As $S_1=baa$, we have $|S_1(b)|_a=1$. For the induction step, let $n > 1$. By the induction hypothesis, we have $|S_{n-1}(b^m)|_a = 2^{n-1-m}$ for $1 \leq m \leq n-1$. As $S_{n-1}$ ends with $a$, we obtain $|S_{n-1}S_{n-1}(b^m)|_a=2^{n-m}$ for $1 \leq m \leq n-1$. Moreover, the construction of $S_n$ implies that the prefix $b^n$ of $S_n$, which is the only occurrence of $b^n$ in $S_n$, is followed by the symbol $a$. Thus, $|S_n(b^m)|_a = 2^{n-m}$ for $1 \leq m \leq n$, which proves the claim. 

As $|S_n(b^m)|_a = 2^{n-m}$, we have $|S_n(b^m)|_b=2^{n-m}-1$. Thus, we obtain the following lower bound for the 
$k^{th}$-order empirical entropy of $S_n$ for $k \in o(n)$. 
\begin{eqnarray*}
H_k(S_n) & = & \sum_{\alpha \in \{a,b\}^k}H(S_n(\alpha)) \\
& \geq & \sum_{x \in \{a,b\}} |S_n(b^k)|_x \log_2 \left(\frac{|S_n(b^k)|}{|S_n(b^k)|_x}\right) \\
&= & 2^{n-k}\log_2 \left(\frac{2^{n-k+1}-1}{2^{n-k}} \right) + (2^{n-k}-1) \log_2\left(\frac{2^{n-k+1}-1}{2^{n-k}-1}\right) \\
& = & (2^{n-k+1}-1) g(2^{n-k}) \\
& \geq & 2^{n-k+1}(1-o(1))
\end{eqnarray*}
This proves the theorem. 
\end{proof}


\begin{appendix}

\section{Histories of length smaller than $k$}\label{sec-padding}

In order to define $k^{th}$-order empirical entropy for binary trees, there are basically three possibilities how to deal with nodes whose history is shorter than $2k$: 
\begin{itemize}
\item[(i)] pad the histories with a fixed dummy symbol $\Box \in \Sigma$ and direction $i \in \{0,1\}$,
\item[(ii)] allow histories of length smaller than $2k$, or, equivalently, pad the histories with a fixed dummy symbol $\diamond \notin \Sigma$ and direction $i \in \{0,1\}$, or
\item[(iii)] ignore nodes whose history is of length smaller than $2k$.
\end{itemize}
Recall that in the main text we used the variant (i) with $i=0$.
In this subsection, we show that the above three variants are basically equivalent if $k$ is small compared to the size of the binary tree. 

Fix an integer $k \geq 1$. Recall that in Section~\ref{sec:tree-entropy} we defined for a tree $t$, a $k$-history $z\in \mathcal{L}_k$, 
and $\tilde{a} \in \Sigma \times \{0,2\}$ the numbers
$m^t_z = |V_z(t)|$ and $m^t_{z,\tilde{a}} = |\{ v \in V_z(t) \mid \lambda(v) = \tilde{a}\}|$.
The tree $t$ will be fixed in this section; hence we will write $m_z$ and $m_{z,\tilde{a}}$ in the following. 
We define several variants of these numbers.

For a $k$-history $z \in \mathcal{L}_k$ and $\tilde{a} \in \Sigma \times \{0,2\}$ we define:
\begin{eqnarray*}
m^{\scriptscriptstyle{<}} &=& |\{ v \in V(t) \mid |v| < k \}|,\\
m_{z}^{\scriptscriptstyle{<}} &=& |\{ v \in V_z(t) \mid |v| < k \}|,\\
m_{z,\tilde{a}}^{\scriptscriptstyle{<}} &=& |\{ v \in V_z(t) \mid |v| < k, \lambda(v) = \tilde{a}\}|, \\
m_{z}^{\scriptscriptstyle{\ge}} &=& |\{ v \in V_z(t) \mid |v| \geq k \}|, \\
m_{z,\tilde{a}}^{\scriptscriptstyle{\ge}} &=& |\{ v \in V_z(t) \mid |v| \geq k, \lambda(v) = \tilde{a}\}| . 
\end{eqnarray*}
We have $m^{\scriptscriptstyle{<}} \leq 2^k-1$
and $m^{\scriptscriptstyle{<}} \geq 2k-1$ if $|t| \geq k$. 
Also note that $m_z = m_{z}^{\scriptscriptstyle{<}} + m_{z}^{\scriptscriptstyle{\ge}}$
and $\sum_{z \in \mathcal{L}_k} m_{z}^{\scriptscriptstyle{<}} = m^{\scriptscriptstyle{<}}$ and $\sum_{z \in \mathcal{L}_k} m_{z}^{\scriptscriptstyle{\ge}} = 2|t|-1-m^{\scriptscriptstyle{<}}$.

Fix a fresh symbol $\diamond \notin \Sigma$ and let
$\mathcal{L}^{\diamond} = ((\Sigma\cup \{\diamond\})\{0,1\})^*$
and $\mathcal{L}_k^{\diamond} = \{w \in \mathcal{L}^{\diamond} \mid |w|=2k\}$. Clearly, $\mathcal{L} \subseteq \mathcal{L}^{\diamond}$ and $\mathcal{L}_k \subseteq \mathcal{L}_k^{\diamond}$.
Let $\ell_k: \mathcal{L}^{\diamond} \rightarrow \mathcal{L}_k^{\diamond}$ denote the partial function mapping a string $z \in \mathcal{L}^{\diamond}$ with $|z| \geq 2k$ to the suffix of $z$ of length $2k$. For a binary tree $t$ and a node $v \in V(t)$, define $h_k^{\diamond}(v) = \ell_k((\diamond 0)^k h(v))$. 
Note that $h_k^{\diamond}(v) = h_k(v)$ for nodes $v \in V(t)$ with $|v| \geq k$.
Finally, for $z \in \mathcal{L}_k^{\diamond}$ and $\tilde{a} \in \Sigma \times \{0,2\}$ we define
\begin{eqnarray*}
m_{z}^{\diamond} &=& |\{ v \in V(t) \mid h_k^{\diamond}(v) = z, |v| < k \}|,\\
m_{z,\tilde{a}}^{\diamond} &=& |\{ v \in V(t) \mid h_k^{\diamond}(v) = z, |v| < k, \lambda(v) = \tilde{a}\}| .
\end{eqnarray*}
Using the above numbers, we can define three natural variations of the $k^{th}$-order empirical entropy
of a binary node-labeled tree $t$:
\begin{itemize}
\item[(i)] Padding histories of length shorter than $2k$ with $\Box \in \Sigma$ and $i \in \{0,1\}$ yields the definition of $k^{th}$-order 
empirical entropy from Section \ref{sec-prelim} (for $i = 0$):
\begin{align*}
H_k(t) = \sum_{z \in \mathcal{L}_k}\sum_{\tilde{a} \in \Sigma \times \{0,2\}}m_{z, \tilde{a}}\log_2\left(\frac{m_{z}}{m_{z,\tilde{a}}}\right).
\end{align*}
\item[(ii)] Padding histories of length shorter than $2k$ with $\diamond \notin \Sigma$ and $i =0$ yields
\begin{align*}
H_k^{\diamond}(t) = \sum_{z \in \mathcal{L}_k^{\diamond}}\sum_{\tilde{a} \in \Sigma \times \{0,2\}}m_{z, \tilde{a}}^{\scriptscriptstyle{\ge}}
\log_2\left(\frac{m_{z}^{\scriptscriptstyle{\ge}}}{m_{z,\tilde{a}}^{\scriptscriptstyle{\ge}}}\right) + m_{z, \tilde{a}}^{\diamond} \log_2 \left(\frac{m_{z}^{\diamond}}{m_{z, \tilde{a}}^{\diamond}}\right).
\end{align*}
This is equivalent to allowing histories of length shorter than $2k$: By padding with a symbol $\diamond \notin \Sigma$, we have $h_k^{\diamond}(v_1) = h_k^{\diamond}(v_2)$ if and only if $h(v_1) = h(v_2)$ for nodes $v_1,v_2 \in V(t)$ with $|v_1|, |v_2| < k$.
\item[(iii)] Ignoring nodes whose history is of length smaller than $2k$ yields
\begin{align*}
H_k^{\scriptscriptstyle{\ge}}(t) = \sum_{z \in \mathcal{L}_k}\sum_{\tilde{a} \in \Sigma \times \{0,2\}}m_{z, \tilde{a}}^{\scriptscriptstyle{\ge}}\log_2\left(\frac{m_{z}^{\scriptscriptstyle{\ge}}}{m_{z,\tilde{a}}^{\scriptscriptstyle{\ge}}}\right).
\end{align*}
\end{itemize}
We can now show that these three approaches are basically equivalent: 
\begin{theorem}\label{theo-diff-entropies}
For every $k \geq 1$ and every binary tree $t$, we have the following:
\begin{eqnarray*}
|H_k(t)-H_k^{\diamond}(t)| & \leq &m^{\scriptscriptstyle{<}}\left(1+\frac{1}{\ln(2)}+\log_2\sigma+\log_2\left(\frac{2|t|-1}{m^{\scriptscriptstyle{<}}}\right)\right), \\
|H_k(t) - H_k^{\scriptscriptstyle{\ge}}(t)| & \leq & m^{\scriptscriptstyle{<}}\left(1+\frac{1}{\ln(2)}+\log_2\sigma+\log_2\left(\frac{2|t|-1}{m^{\scriptscriptstyle{<}}}\right)\right), \\
|H_k^{\scriptscriptstyle{\ge}}(t) - H_k^{\diamond}(t)| & \leq & m^{\scriptscriptstyle{<}}(1+\log_2\sigma) .
\end{eqnarray*}
\end{theorem}

\begin{proof}
First, note that 
\begin{align}\label{estimate-nodes}
0 \leq  \sum_{z \in \mathcal{L}_k}m_{z}^{\scriptscriptstyle{<}} \sum_{\tilde{a} \in \Sigma \times \{0,2\}}\frac{m_{z,\tilde{a}}^{\scriptscriptstyle{<}}}{m_{z}^{\scriptscriptstyle{<}}}\log_2\left(\frac{m_{z}^{\scriptscriptstyle{<}}}{m_{z,\tilde{a}}^{\scriptscriptstyle{<}}}\right) \leq m^{\scriptscriptstyle{<}}(1+\log_2\sigma),
\end{align}
as the inner sum is the Shannon entropy $H(P)$ of the probability distribution $P: \Sigma \times \{0,2\} \rightarrow [0,1]$ given by $P(\tilde{a})=m_{z,\tilde{a}}^{\scriptscriptstyle{<}}/m_z^{\scriptscriptstyle{<}}$ (and hence $H(P) \leq \log_2(2\sigma) = 1+\log_2 \sigma$)
 and as $\sum_{z \in \mathcal{L}_k}m_{z}^{\scriptscriptstyle{<}} = m^{\scriptscriptstyle{<}}$. Analogously, we get
\begin{align}\label{estimate-nodes2}
0 \leq  \sum_{z \in \mathcal{L}_k^{\diamond}}m_z^{\diamond} \sum_{\tilde{a} \in \Sigma \times \{0,2\}}\frac{m_{z, \tilde{a}}^{\diamond}}{m_z^{\diamond}}\log_2\left(\frac{m_z^{\diamond}}{m_{z, \tilde{a}}^{\diamond}}\right) \leq m^{\scriptscriptstyle{<}}(1+\log_2\sigma).
\end{align}
We start with upper-bounding $|H_k(t) - H_k^{\scriptscriptstyle{\ge}}(t)|$:
By the log-sum inequality (Lemma~\ref{logsum}) and \eqref{estimate-nodes}, we get
\begin{align*}
H_k(t) &= \sum_{z \in \mathcal{L}_k}\sum_{\tilde{a} \in \Sigma \times \{0,2\}}m_{z, \tilde{a}}\log_2\left(\frac{m_{z}}{m_{z,\tilde{a}}}\right) \\
&= \sum_{z \in \mathcal{L}_k}\sum_{\tilde{a} \in \Sigma \times \{0,2\}}(m_{z, \tilde{a}}^{\scriptscriptstyle{\ge}}+m_{z,\tilde{a}}^{\scriptscriptstyle{<}})\log_2\left(\frac{m_{z}^{\scriptscriptstyle{\ge}}+m_{z}^{\scriptscriptstyle{<}}}{m_{z,\tilde{a}}^{\scriptscriptstyle{\ge}}+m_{z,\tilde{a}}^{\scriptscriptstyle{<}}}\right) \\
&\geq \sum_{z \in \mathcal{L}_k}\sum_{\tilde{a} \in \Sigma \times \{0,2\}}m_{z, \tilde{a}}^{\scriptscriptstyle{\ge}}\log_2\left(\frac{m_{z}^{\scriptscriptstyle{\ge}}}{m_{z,\tilde{a}}^{\scriptscriptstyle{\ge}}}\right) + \sum_{z \in \mathcal{L}_k}\sum_{\tilde{a} \in \Sigma \times \{0,2\}}m_{z,\tilde{a}}^{\scriptscriptstyle{<}}\log_2\left(\frac{m_{z}^{\scriptscriptstyle{<}}}{m_{z,\tilde{a}}^{\scriptscriptstyle{<}}}\right)\\
&\geq \sum_{z \in \mathcal{L}_k}\sum_{\tilde{a} \in \Sigma \times \{0,2\}}m_{z, \tilde{a}}^{\scriptscriptstyle{\ge}}\log_2\left(\frac{m_{z}^{\scriptscriptstyle{\ge}}}{m_{z,\tilde{a}}^{\scriptscriptstyle{\ge}}}\right) \\
& = H_k^{\scriptscriptstyle{\ge}}(t).
\end{align*}
Moreover, we find
\begin{align*}
H_k(t) = & \sum_{z \in \mathcal{L}_k}\sum_{\tilde{a} \in \Sigma \times \{0,2\}}(m_{z, \tilde{a}}^{\scriptscriptstyle{\ge}}+m_{z,\tilde{a}}^{\scriptscriptstyle{<}})\log_2\left(\frac{m_{z}^{\scriptscriptstyle{\ge}}+m_{z}^{\scriptscriptstyle{<}}}{m_{z,\tilde{a}}^{\scriptscriptstyle{\ge}}+m_{z,\tilde{a}}^{\scriptscriptstyle{<}}}\right)\\
= &\sum_{z \in \mathcal{L}_k}\sum_{\tilde{a} \in \Sigma \times \{0,2\}}m_{z, \tilde{a}}^{\scriptscriptstyle{\ge}}\log_2\left(\frac{m_{z}^{\scriptscriptstyle{\ge}}+m_{z}^{\scriptscriptstyle{<}}}{m_{z}^{\scriptscriptstyle{\ge}}}\cdot \frac{m_{z}^{\scriptscriptstyle{\ge}}}{m_{z,\tilde{a}}^{\scriptscriptstyle{\ge}}+m_{z,\tilde{a}}^{\scriptscriptstyle{<}}}\right) + \\
&\sum_{z \in \mathcal{L}_k}\sum_{\tilde{a} \in \Sigma \times \{0,2\}}m_{z, \tilde{a}}^{\scriptscriptstyle{<}}\log_2\left(\frac{m_{z}^{\scriptscriptstyle{\ge}}+m_{z}^{\scriptscriptstyle{<}}}{m_{z}^{\scriptscriptstyle{<}}} \cdot \frac{m_{z}^{\scriptscriptstyle{<}}}{m_{z,\tilde{a}}^{\scriptscriptstyle{\ge}}+m_{z,\tilde{a}}^{\scriptscriptstyle{<}}}\right)\\
\leq &\sum_{z \in \mathcal{L}_k}m_{z}^{\scriptscriptstyle{\ge}}\log_2\left(\frac{m_{z}^{\scriptscriptstyle{\ge}}+m_{z}^{\scriptscriptstyle{<}}}{m_{z}^{\scriptscriptstyle{\ge}}}\right)
+\sum_{z \in \mathcal{L}_k}\sum_{\tilde{a} \in \Sigma \times \{0,2\}}m_{z, \tilde{a}}^{\scriptscriptstyle{\ge}}\log_2\left( \frac{m_{z}^{\scriptscriptstyle{\ge}}}{m_{z,\tilde{a}}^{\scriptscriptstyle{\ge}}}\right)+\\
&\sum_{z \in \mathcal{L}_k}m_{z}^{\scriptscriptstyle{<}}\log_2\left(\frac{m_{z}^{\scriptscriptstyle{\ge}}+m_{z}^{\scriptscriptstyle{<}}}{m_{z}^{\scriptscriptstyle{<}}}\right)
+\sum_{z \in \mathcal{L}_k}\sum_{\tilde{a} \in \Sigma \times \{0,2\}}m_{z, \tilde{a}}^{\scriptscriptstyle{<}}\log_2\left( \frac{m_{z}^{\scriptscriptstyle{<}}}{m_{z,\tilde{a}}^{\scriptscriptstyle{<}}}\right)\\
\leq & \; H_k^{\scriptscriptstyle{\ge}}(t)+m^{\scriptscriptstyle{<}}\left(1+\log_2\sigma\right)+\\ 
& \; \left(2|t|-1-m^{\scriptscriptstyle{<}}\right)\log_2\left(\frac{2|t|-1}{2|t|-1-m^{\scriptscriptstyle{<}}}\right)+m^{\scriptscriptstyle{<}}\log_2\left(\frac{2|t|-1}{m^{\scriptscriptstyle{<}}}\right)
\end{align*}
by the log-sum inequality (Lemma~\ref{logsum}) and our estimate from (\ref{estimate-nodes}). We have
\begin{align}\label{meanvaluetheorem}
\left(2|t|-1-m^{\scriptscriptstyle{<}}\right)\log_2\left(\frac{2|t|-1}{2|t|-1-m^{\scriptscriptstyle{<}}}\right) \leq \frac{m^{\scriptscriptstyle{<}}}{\ln(2)},
\end{align}
which follows immediately from the mean-value theorem: as a consequence of the mean-value theorem, for every mapping $f: [a,b] \rightarrow \mathbb{R}$, which is differentiable on $[a,b]$, we have 
\begin{align*}
|f(b) - f(a)| \leq \max_{x \in [a,b]} |f'(x)| \cdot |b-a|.
\end{align*}
With $f(x) = \log_2(x)$, $a =2|t|-1-m^{\scriptscriptstyle{<}}$ and $b = 2|t|-1$ and by logarithmic identities, we obtain the estimate \eqref{meanvaluetheorem}. 
Thus, we have:
$$|H_k(t) - H_k^{\scriptscriptstyle{\ge}}(t)| \leq  m^{\scriptscriptstyle{<}}\left(1+\log_2\sigma+\frac{1}{\ln(2)} + \log_2\left(\frac{2|t|-1}{m^{\scriptscriptstyle{<}}}\right)\right).$$ 
Next, we upper-bound $|H_k^{\scriptscriptstyle{\ge}}(t) - H_k^{\diamond}(t)|$: From the definitions of $H_k^{\scriptscriptstyle{\ge}}(t)$ and  $H_k^{\diamond}(t)$, we get
\begin{align*}
H_k^{\diamond}(t) = H_k^{\scriptscriptstyle{\ge}}(t) + \sum_{z \in \mathcal{L}_k^{\diamond}} \sum_{\tilde{a} \in \Sigma \times \{0,2\}}m_{z, \tilde{a}}^{\diamond}\log_2\left(\frac{m_z^{\diamond}}{m_{z, \tilde{a}}^{\diamond}}\right).
\end{align*} 
As the second sum on the right-hand side is between $0$ and $m^{\scriptscriptstyle{<}}(1+\log_2\sigma)$ (see (\ref{estimate-nodes2})), 
we get $|H_k^{\scriptscriptstyle{\ge}}(t) - H_k^{\diamond}(t)|  \leq  m^{\scriptscriptstyle{<}}(1+\log_2\sigma)$.

Finally, as $H_k^{\diamond}(t)\geq H_k^{\scriptscriptstyle{\ge}}(t)$ and $H_k(t)\geq H_k^{\scriptscriptstyle{\ge}}(t)$, we have
\begin{align*}
|H_k(t)-H_k^{\diamond}(t)|&\leq
 m^{\scriptscriptstyle{<}}\left(1+\log_2\sigma+\frac{1}{\ln(2)}+\log_2\left(\frac{2|t|-1}{m^{\scriptscriptstyle{<}}}\right)\right).
\end{align*}
This proves the theorem.
\end{proof}
Theorem~\ref{theo-diff-entropies} moreover shows that the choice of the symbol $\Box \in \Sigma$ used for padding the histories only affects the value of the $k^{th}$-order empirical entropy by an additive term of at most $m^{\scriptscriptstyle{<}}(1+\log_2\sigma+1/\ln(2))+m^{\scriptscriptstyle{<}}\log_2((2|t|-1)/m^{\scriptscriptstyle{<}})$.

\section{Proof of Lemma~\ref{lemma-H_2k}} \label{appendix-H_2k}

Fix a binary tree $t \in \mathcal{T}(\Sigma)$.
By definition of the first-child next-sibling encoding, every inner node of $\ffcns(t)$ corresponds in a bijective manner to a node of $t$: For an inner node $v$ of $\ffcns(t)$, let $\ffcns^{-1}(v)$ denote the corresponding node of $t$ and let $\ffcns(v)$ denote the corresponding inner node of $\ffcns(t)$ of a node $v$ of $t$. If $v$ is a node of $t$, then we obtain $h(\ffcns(v))$ as follows: If $v = \varepsilon$, then $h(\ffcns(v)) = \varepsilon$. Moreover, if $v$ is a left child of a node $\parent(v)$ with label $a \in \Sigma$, then $h(\ffcns(v))=h(\ffcns(\parent(v)))a0$ (and $h(v)=h(\parent(v))a0$). Finally, if $v$ is a right child of a node $\parent(v)$ with label $a \in \Sigma$ and $v$'s left sibling has label $a' \in \Sigma$, then $h(\ffcns(v))=h(\ffcns(\parent(v)))a0a'1$ (and $h(v)=h(\parent(v))a1$). Thus, we are also able to determine $h(\ffcns^{-1}(v))$ from $h(v)$ for every inner node $v$ of $\ffcns(t)$: locating every occurrence of a pattern of the form $0a1$ with $a \in \Sigma$ in the string $h(v)$ and replacing it by $1$ yields $h(\ffcns^{-1}(v))$. 

In particular,  we have $|h(\ffcns(v))|\leq 2|h(v)|$ for every node $v$ of $t$, respectively, $|h(\ffcns^{-1}(v))|\geq 1/2|h(v)|$ for every inner node $v$ of $\ffcns(t)$. Moreover, for every inner node $v$ of $\ffcns(t)$, we can uniquely determine $h_k(\ffcns^{-1}(v))$ from $h_{2k}(v)$. Thus, we are also able to determine $h_{k-1}(\ffcns^{-1}(v))$ from $h_{2k-1}(v)$ for every inner node $v$ of $\ffcns(t)$.
Let 
\begin{align*}
\mathcal{L}_{m}(\ffcns(t)) = \{h_{m}(v) \mid v \text{ is an inner node of } \ffcns(t)\}
\end{align*}
denote the set of $m$-histories that appear as $m$-history of an inner node of $\ffcns(t)$. 
We define a mapping $\varphi: \mathcal{L}_{2k}(\ffcns(t)) \rightarrow\mathcal{L}_k$ by $\varphi(h_{2k}(v))=h_k(\ffcns^{-1}(v))$, which maps the $2k$-history of an inner node of $\ffcns(t)$ to the $k$-history of the corresponding node in $t$: By the above considerations, this mapping is well-defined. 
Furthermore, we define a mapping $\pi: \mathcal{L}_{2k-1}(\ffcns(t)) \rightarrow \mathcal{\mathcal{L}}_{k-1}$ by $\pi(h_{2k-1}(v))=h_{k-1}(\ffcns^{-1}(v))$. Again, by the above considerations, this mapping is well-defined, as we are able to determine $h_{k-1}(\ffcns^{-1}(v))$ from $h_{2k-1}(v)$. 

For $m \geq 2$ we partition $\mathcal{L}_{m}$ into the following disjoint subsets:
\begin{align*}
 \mathcal{L}_{m}^0&=\{a_1i_1\cdots a_{m}i_{m}  \in  \mathcal{L}_{m} \mid i_{m}=0\},\\
  \mathcal{L}_{m}^{01}&=\{a_1i_1\cdots a_{m}i_{m}  \in  \mathcal{L}_{m} \mid i_{m-1}=0 \text{ and } i_{m}=1\},\\
  \mathcal{L}_{m}^{11}&=\{a_1i_1\cdots a_{m}i_{m}  \in  \mathcal{L}_{m} \mid i_{m-1}=1 \text{ and } i_{m}=1\}.
\end{align*}
Moreover, we define $\mathcal{L}_{2k}^s(\ffcns(t))=\mathcal{L}_{2k}^s \cap \mathcal{L}_{2k}(\ffcns(t))$ for $s \in \{0,01,11\}$.
We observe the following:
\begin{itemize}
\item[(i)] If $h_{2k}(v) \in  \mathcal{L}_{2k}^{11}$ for a node $v$ of $\ffcns(t)$, then $v$ is a $\Box$-labeled leaf of $\ffcns(t)$: As $t$ is a binary tree, the right sibling of a node has no right
sibling. Thus, there are no inner nodes $v$ in $\ffcns(t)$ with $h_{2k}(v) \in  \mathcal{L}_{2k}^{11}$.
\item[(ii)] If $h_{2k}(v) \in  \mathcal{L}_{2k}^{01}$ for a node $v$ of $\ffcns(t)$, then $v$ is an inner node of $\ffcns(t)$: This follows again 
from the fact that $t$ is a binary tree (and hence does not have unary nodes). 
\item[(iii)] If $h_{2k}(v) \in  \mathcal{L}_{2k}^{0}$ for a node $v$ of $\ffcns(t)$, then $v$ can be an inner node or a leaf of $\ffcns(t)$. 
If $v$ is a leaf, then its label is the fixed dummy symbol $\Box \in \Sigma$. 
\item[(iv)] For every $i \in \{0,1\}$ and node $v$ of $t$,
we have $h_k(v) \in \mathcal{L}_k^i$ if and only if $h_{2k}(\ffcns(v)) \in  \mathcal{L}_{2k}^i(\ffcns(t))$. 
In particular $\varphi(z) \in \mathcal{L}_k^0$ for every $z \in \mathcal{L}_{2k}^0(\ffcns(t))$ and $\varphi(z) \in \mathcal{L}_k^{1}$ for every $z \in \mathcal{L}_{2k}^{01}(\ffcns(t))$. Hence 
$\varphi(z) \neq \varphi(z')$ if $z \in \mathcal{L}_{2k}^{01}(\ffcns(t))$ and $z' \in \mathcal{L}_{2k}^{0}(\ffcns(t))$. 
\end{itemize}
From (i), we obtain
\begin{align}\label{eq-L2k11}
\sum_{z \in \mathcal{L}_{2k}^{11}} \sum_{\tilde{a} \in \Sigma \times \{0,2\}} m_{z,\tilde{a}}^{\ffcns(t)}\log_2\left(\frac{m_{z}^{\ffcns(t)}}{m_{z,\tilde{a}}^{\ffcns(t)}}\right)=0.
\end{align}
From (ii) and (iv), we obtain the following:
\begin{eqnarray*}
 & & \sum_{z \in \mathcal{L}_{2k}^{01}} \sum_{\tilde{a} \in \Sigma \times \{0,2\}} m_{z,\tilde{a}}^{\ffcns(t)}\log_2\left(\frac{m_{z}^{\ffcns(t)}}{m_{z,\tilde{a}}^{\ffcns(t)}}\right) \\
& = & \sum_{z \in \mathcal{L}_{2k}^{01}(\ffcns(t))} \sum_{a \in \Sigma} m_{z,(a,2)}^{\ffcns(t)}\log_2\left(\frac{m_{z}^{\ffcns(t)}}{m_{z,(a,2)}^{\ffcns(t)}}\right)\\
& \leq & \sum_{y \in \mathcal{L}_k^1}  \sum_{a \in \Sigma} \left(\sum_{z \in \varphi^{-1}(y)}m_{z,(a,2)}^{\ffcns(t)}\right)\log_2\left(\frac{\sum_{z \in \varphi^{-1}(y)}m_{z}^{\ffcns(t)}}{\sum_{z \in \varphi^{-1}(y)}m_{z,(a,2)}^{\ffcns(t)}}\right),
\end{eqnarray*}
where the last estimate follows from the log-sum inequality (Lemma~\ref{logsum}). For every $y \in \mathcal{L}_k^1$ we have
\begin{align*}
&\sum_{z \in \varphi^{-1}(y)}m_{z,(a,2)}^{\ffcns(t)}=m_{y,(a,0)}^t + m_{y,(a,2)}^t, \\
&\sum_{z \in \varphi^{-1}(y)}m_{z}^{\ffcns(t)}=m_{y}^t.
\end{align*}
Thus, we obtain
\begin{eqnarray}\label{eq-Lk01}
& & \sum_{z \in \mathcal{L}_{2k}^{01}} \sum_{\tilde{a} \in \Sigma \times \{0,2\}} m_{z,\tilde{a}}^{\ffcns(t)}\log_2\left(\frac{m_{z}^{\ffcns(t)}}{m_{z,\tilde{a}}^{\ffcns(t)}}\right) \notag \\
& \leq & \sum_{y \in \mathcal{L}_k^1}\sum_{a \in \Sigma}\left(m_{y,(a,0)}^t + m_{y,(a,2)}^t\right)\log_2\left(\frac{m_{y}^t}{m_{y,(a,0)}^t + m_{y,(a,2)}^t}\right).
\end{eqnarray}
From (iii) and (iv), we obtain 
\begin{eqnarray*}
& & \sum_{z \in \mathcal{L}_{2k}^{0}} \sum_{\tilde{a} \in \Sigma \times \{0,2\}} m_{z,\tilde{a}}^{\ffcns(t)}\log_2\left(\frac{m_{z}^{\ffcns(t)}}{m_{z,\tilde{a}}^{\ffcns(t)}}\right)\\
&=& {\sum_{z \in \mathcal{L}_{2k}^{0}(\ffcns(t))} \sum_{a \in \Sigma} m_{z,(a,2)}^{\ffcns(t)}\log_2\left(\frac{m_{z}^{\ffcns(t)}}{m_{z,(a,2)}^{\ffcns(t)}}\right)}+\sum_{z \in \mathcal{L}_{2k}^{0}}m_{z,(\Box,0)}^{\ffcns(t)}\log_2\left(\frac{m_{z}^{\ffcns(t)}}{m_{z,(\Box,0)}^{\ffcns(t)}}\right).
\end{eqnarray*}
For the first summand, we find analogously as in the previous estimate (\ref{eq-Lk01}):
\begin{eqnarray}\label{eq-L2k0a}
& & {\sum_{z \in \mathcal{L}_{2k}^{0}(\ffcns(t))} \sum_{a \in \Sigma} m_{z,(a,2)}^{\ffcns(t)}\log_2\left(\frac{m_{z}^{\ffcns(t)}}{m_{z,(a,2)}^{\ffcns(t)}}\right)} \notag \\
& \leq & \sum_{y \in \mathcal{L}_k^0}\sum_{a \in \Sigma}\left(m_{y,(a,0)}^t + m_{y,(a,2)}^t\right)\log_2\left(\frac{m_{y}^t}{m_{y,(a,0)}^t + m_{y,(a,2)}^t}\right).
\end{eqnarray}
For the second summand, we obtain as $k \geq 1$:
\begin{align*}
&\sum_{z \in \mathcal{L}_{2k}^{0}}m_{z,(\Box,0)}^{\ffcns(t)}\log_2\left(\frac{m_{z}^{\ffcns(t)}}{m_{z,(\Box,0)}^{\ffcns(t)}}\right) \\
 =& \sum_{z \in \mathcal{L}_{2k-1}}\sum_{a \in \Sigma}m_{za0,(\Box,0)}^{\ffcns(t)}\log_2\left(\frac{m_{za0}^{\ffcns(t)}}{m_{za0,(\Box,0)}^{\ffcns(t)}}\right)\\
\leq & \sum_{y \in \mathcal{L}_{k-1}} \sum_{a \in \Sigma}\left(\sum_{z \in \pi^{-1}(y)}m_{za0,(\Box,0)}^{\ffcns(t)}\right)\log_2\left(\frac{\sum_{z \in \pi^{-1}(y)}m_{za0}^{\ffcns(t)}}{\sum_{z \in \pi^{-1}(y)}m_{za0,(\Box,0)}^{\ffcns(t)}}\right),
\end{align*}
where the last inequality follows from the log-sum inequality. Moreover, for all $y \in \mathcal{L}_{k-1}$ we have
\begin{align*}
&\sum_{z \in \pi^{-1}(y)}m_{za0,(\Box,0)}^{\ffcns(t)}=m_{y,(a,0)}^t,\\
&\sum_{z \in \pi^{-1}(y)}m_{za0}^{\ffcns(t)}=m_{y,(a,0)}^t+m_{y,(a,2)}^t.
\end{align*}
Thus, we find
\begin{eqnarray}\label{eq-L2k0b}
& & \sum_{z \in \mathcal{L}_{2k}^{0}}m_{z,(\Box,0)}^{\ffcns(t)}\log_2\left(\frac{m_{z}^{\ffcns(t)}}{m_{z,(\Box,0)}^{\ffcns(t)}}\right) \notag \\
& \leq & \sum_{y \in \mathcal{L}_{k-1}}\sum_{a \in \Sigma}m_{y,(a,0)}^t\log_2\left(\frac{m_{y,(a,0)}^t+m_{y,(a,2)}^t}{m_{y,(a,0)}^t}\right). 
\end{eqnarray}
Altogether, if we combine the estimates from (\ref{eq-L2k11}), (\ref{eq-Lk01}), (\ref{eq-L2k0a}) and (\ref{eq-L2k0b}), we obtain:
\begin{align*}
&H_{2k}(\ffcns(t)) = \sum_{z \in \mathcal{L}_{2k}}\sum_{\tilde{a} \in \Sigma \times \{0,2\}}m_{z, \tilde{a}}^{\ffcns(t)}\log_2\left(\frac{m_z^{\ffcns(t)}}{m_{z, \tilde{a}}^{\ffcns(t)}}\right)\\
&= \sum_{z \in \mathcal{L}_{2k}^0}\sum_{\tilde{a} \in \Sigma \times \{0,2\}}m_{z, \tilde{a}}^{\ffcns(t)}\log_2\left(\frac{m_z^{\ffcns(t)}}{m_{z, \tilde{a}}^{\ffcns(t)}}\right)+\sum_{z \in \mathcal{L}_{2k}^{01}}\sum_{\tilde{a} \in \Sigma \times \{0,2\}}m_{z, \tilde{a}}^{\ffcns(t)}\log_2\left(\frac{m_z^{\ffcns(t)}}{m_{z, \tilde{a}}^{\ffcns(t)}}\right)\\&+
\sum_{z \in \mathcal{L}_{2k}^{11}}\sum_{\tilde{a} \in \Sigma \times \{0,2\}}m_{z, \tilde{a}}^{\ffcns(t)}\log_2\left(\frac{m_z^{\ffcns(t)}}{m_{z, \tilde{a}}^{\ffcns(t)}}\right)\\
&\leq \sum_{z \in \mathcal{L}_k^1}\sum_{a \in \Sigma}\left(m_{z,(a,0)}^t + m_{z,(a,2)}^t\right)\log_2\left(\frac{m_{z}^t}{m_{z,(a,0)}^t + m_{z,(a,2)}^t}\right)\\
&+\sum_{z \in \mathcal{L}_k^0}\sum_{a \in \Sigma}\left(m_{z,(a,0)}^t + m_{z,(a,2)}^t\right)\log_2\left(\frac{m_{z}^t}{m_{z,(a,0)}^t + m_{z,(a,2)}^t}\right)\\
&+\sum_{z \in \mathcal{L}_{k-1}}\sum_{a \in \Sigma}m_{z,(a,0)}^t\log_2\left(\frac{m_{z,(a,0)}^t+m_{z,(a,2)}^t}{m_{z,(a,0)}^t}\right)\\
&\leq \sum_{z \in \mathcal{L}_{k-1}}\sum_{a \in \Sigma}\left(m_{z,(a,0)}^t + m_{z,(a,2)}^t\right)\log_2\left(\frac{m_{z}^t}{m_{z,(a,0)}^t + m_{z,(a,2)}^t}\right)\\
&+\sum_{z \in \mathcal{L}_{k-1}}\sum_{a \in \Sigma}m_{z,(a,0)}^t\log_2\left(\frac{m_{z,(a,0)}^t+m_{z,(a,2)}^t}{m_{z,(a,0)}^t}\right)\\
&\leq \sum_{z \in \mathcal{L}_{k-1}}\sum_{\tilde{a} \in \Sigma\times\{0,2\}}m_{z, \tilde{a}}^t\log_2\left(\frac{m_z^t}{m_{z, \tilde{a}}^t}\right)=H_{k-1}(t),
\end{align*}
where the last-but-one estimate follows again from the log-sum inequality.
This proves Lemma~\ref{lemma-H_2k}.
\qed

\end{appendix}

\end{document}